\documentclass[lettersize,journal]{IEEEtran}
\usepackage{amsmath,amsfonts}
\usepackage{algorithmic}
\usepackage{algorithm}
\usepackage{array}
\usepackage[caption=false,font=normalsize,labelfont=sf,textfont=sf]{subfig}
\usepackage{textcomp}
\usepackage{stfloats}
\usepackage{url}
\usepackage{verbatim}
\usepackage{graphicx}
\usepackage{cite}
\usepackage{amssymb}
\usepackage{amsthm}
\usepackage{flushend}
\usepackage{float}
\usepackage{subcaption}
\newtheorem{lemma}{Lemma}
\newtheorem{theorem}{Theorem}
\begin{document}

\title{Distributed Forgetting-factor Regret-based Online Optimization over Undirected Connected Networks}

\author{Lipo Mo, Jianjun Li, Min Zuo, \IEEEmembership{Member,~IEEE }, Lei Wang
\thanks{The authors are with the Institute of Systems Science, Beijing Wuzi University, Beijing 101126, P. R. China, also with School of Mathematics and Statistics, Beijing Technology and Business University, Beijing 100048, P. R. China, also with School of Automation Science and Electrical Engineering, Beihang University, Beijing 100191, P. R. China (e-mail: lwang@buaa.edu.cn; beihangmlp@126.com; jianjun2021@sina.cn; zuominbtbu@126.com)}.
\thanks{This work is supported by the National Natural Science Foundation of China under Grant No. 62473009.}}

\markboth{Journal of \LaTeX\ Class Files,~Vol.~14, No.~8, August~2021}%
{Shell \MakeLowercase{\textit{et al.}}: A Sample Article Using IEEEtran.cls for IEEE Journals}


\maketitle

\begin{abstract}
The evaluation of final-iteration tracking performance is a formidable obstacle in distributed online optimization algorithms. To address this issue, this paper proposes a novel evaluation metric named distributed forgetting-factor regret (DFFR). It incorporates a weight into the loss function at each iteration, which progressively reduces the weights of historical loss functions while enabling dynamic weights allocation across optimization horizon. Furthermore, we develop two distributed online optimization algorithms based on DFFR over undirected connected networks: the Distributed Online Gradient-free Algorithm for bandit-feedback problems and the Distributed Online Projection-free Algorithm for high-dimensional problems. Through theoretical analysis, we derive the upper bounds of DFFR for both algorithms and further prove that under mild conditions, DFFR either converges to zero or maintains a tight upper bound as iterations approach infinity. Experimental simulation demonstrates the effectiveness of the algorithms and the superior performance of DFFR.
\end{abstract}


\section{Introduction}
\IEEEPARstart{O}{nline} optimization is a versatile technique with widespread applications across various domains, such as communication network \cite{1,2}, grid systems \cite{a5, a6, a7}, energy management \cite{a9}, where dynamic decision-making is crucial.
The online optimization problem involves finding a series of decisions to minimize cumulative objective functions, which are dynamic and uncertain compared to offline optimization problems.
Recently, many centralized online algorithms have been proposed to address this challenge \cite{5,6,7,8,9,10,11,12,13,14}. For unconstrained online optimization, several online convex optimization algorithms have been developed based on gradient descent method, prediction-correction method and randomization technique \cite{5,6,7,8}. For constrained situation, some robust algorithms, such as online gradient algorithm and online mirror descent algorithm, have been introduced to find the optimal decision sequence \cite{9,10,12,13,14}.  The common performance metric for online optimization algorithms is regret, first introduced by \cite{a12}. The algorithm is called to have good performance if its regret is sublinear, which means that the average loss of the algorithm, in hindsight, approaches the average loss of the best solution. Although the regret of the algorithms mentioned above can reach sublinearity, it cannot always guarantee that the decision generated at the final iteration is close to the optimal solution.  In order to solve this problem, the concept of forgetting-factor regret (FFR) was introduced, which incorporated a forgetting-factor (FF) to reduce the weight of past loss functions and analyzed the tracking performance at the final iteration \cite{a14}. However, when the objective function is very complex, such as the sum of some local objective functions, the computing amount becomes very large and the centralized online optimization algorithms mentioned above are ineffective. 

Distributed optimization algorithms offer significant advantages in handling large-scale, high-computation online optimization problems when compared to centralized online optimization methods. 
These algorithms typically involve multi-agent systems (MAS), which consist of multiple autonomous agents that interact and collaborate to achieve shared goals. In a MAS, each agent usually operates based on its own local information and decision-making rules. However, in many practical scenarios, such as smart grids \cite{a5,a6,a7} or drone swarms \cite{219}, agents need to coordinate with others to optimize collective performance, often through communication networks. The need for online optimization in such systems arises from the dynamic and real-time nature of these applications. For example, in the smart grid, the agent must constantly adjust its operation according to the changing energy demand and supply, while considering the current state of the grid. Similarly, in drone swarms, agents must make real-time decisions to achieve coordinated behaviors, such as path planning or task allocation, while adapting to environmental changes or mission requirements. In these dynamic and unpredictable environments, online optimization allows agents to adjust their strategies continuously, ensuring optimal system performance even in the face of evolving network topologies or incomplete information. Thus, online optimization provides a flexible and adaptive framework to handle the complexities of MAS in real-world applications. 

Recently, several distributed online optimization algorithms have been proposed within MAS \cite{18,19,20,21,22,23,24,25}. For example, the full information feedback algorithms were proposed with static and dynamic regrets, where each agent can access the gradient information of the local objective function . In other works, distributed online algorithms have been designed with event-triggered mechanisms to reduce communication overhead in scenarios with time-varying constraints and time-delaying \cite{24,25}. It is worth noting that most of the algorithms for constrained online optimization problems rely on projection operators, which tend to have high computational complexity. To reduce the gradient computing, some algorithms, such as mirror descent method, stochastic gradient method and Frank-Wolfe method, have been proposed  \cite{26,27,29}. Although these distributed online optimization algorithms achieve sublinear regret, their tracking performance at final iteration  has not been analyzed. This presents a challenge in ensuring that the decisions made at the final iteration are sufficiently close to the optimal solution. To address this issue, we proposed  DFFR by incorporating a weight into the loss function at each iteration. For the full information feedback, we designed Distributed Online Gradient Decent Algorithm based on projection operator in \cite{m}, assuming that the gradient of the local objective function is accessible. In fact, the computation of gradients and projection operators might bring high load to the algorithm, which would limit the application of the distributed algorithms to real-world problems, such as resource allocation. 

To the best of our knowledge, there is no work on DFFR  for distributed algorithms when the gradient can not be accessed or the projection operator can not be utilized. For this problem, this paper mainly propose two algorithms for distinct distributed online optimization problems over undirected connected networks, where the gradient of local objective function can not be accessed or the projection operator can not be utilized.  The primary contributions of this paper are as follows.

(1) We propose gradient-free and projection-free algorithms for bandit-feedback and high-dimensional problems. Compared to \cite{m}, which uses gradients and projection operators, our first algorithm eliminates gradient information, reducing communication overhead, while our second algorithm simplifies optimization progress and improves computational efficiency. Unlike \cite{aa,bb}, where the gradient-free and projection-free algorithms were designed but loss function was assumed to be strongly convex, we consider general convex functions. This broader conditions makes our algorithms more widely applicable in practical scenarios.

(2) We design DFFR by introducing the FF to regret, which can guarantee satisfactory tracking performance. Based on the proposed theoretical framework, we establish the upper bounds of DFFR for the proposed algorithms and derive mild conditions that guarantee the DFFR either converges to zero or remains within a tight upper bound as the number of iterations approaches infinity. Compared with works \cite{18,20}, where the proposed algorithms were sublinear, our method achieves tracking performance at final iteration with tight bounds, which is not satisfied by sublinear.  \par

The remaining parts are as follows. In section 2, we provide a detailed description of the problem, assumptions and some necessary lemmas. In section 3, 4, we introduce Distributed Online Gradient-free Algorithm and Distributed Online Projection-free Algorithm respectively, then establish the upper bounds of their DFFR. In section 5, we provide some numerical simulations to demonstrate our results.

\emph{Notion:} ${\mathbb{R}^{d}}$ represents the set of $d$-dimensional vectors. ${\mathbb{N}_{+}}$ denotes the positive integer set.  Define ${P_{X}}(y)\triangleq \arg {{\min }_{x\in X}}\left\| x-\left. y \right\| \right.^2$ as the projection of $y$ onto the set X.
$\left\langle \cdot ,\cdot  \right\rangle $ represents the inner product of the vector. $\left\| \cdot  \right\|$ represents the norm of some vector.

\section{Problem description and lemmas}
Consider a MAS with $n$ agents, where each agent can communicate with its neighbors over a graph $G$. The adjacency matrix composed of the graph is $W$. Among them, the weight associated between agent $i$ and $j$ is $\omega_{ij}$. Let $\left[ n \right] = \left\{ 1,2,...,n \right\}$. The purpose of  this paper is to design a series of distributed algorithms to find a decision sequence $\left\{ x_{i}^{t},t=1,\cdot \cdot \cdot ,T , i=1,\cdot \cdot \cdot n\right\}$ ($x_i^t$ means the decision generated by agent $i$ at iteration time $t$) for each agent minimizing the following entire objective functions over a finite duration $T \ge 1$:
\begin{align}
	&\min \sum\limits_{t=1}^{T}{{{f}_{t}}\left( x \right)} \\
	&s.t. x \in X,\nonumber
\end{align}
where $X \subset \mathbb{R}^{d}$, 	${{f}_{t}}\left( x \right)\triangleq \frac{1}{n}\underset{i=1}{\overset{n}{\mathop \sum }}\,f_{i}^{t}\left( x \right)$, and $f_i^t: X \to \mathbb{R}$ is the local objective function. Here, once the agent $i$ makes the decision, the local objective function $f_i^t$ is revealed. Then the agent incurs an objective function value $f_i^t \left(x_i^t\right)$. Subsequently, a decision sequence $\left\{ x_{i}^{t} \right\}$ of the local objective $\left\{ f_i^t\right\}$ is obtained. Given the FF $\rho \in (0,1)$, the DFFR of the corresponding algorithm is 
\begin{align}
	R_{T}^{F}\text{ }\!\!~\!\!\text{ }\triangleq \underset{t=1}{\overset{T}{\mathop \sum }}\,{{\rho }^{T-t}}\left[ \frac{1}{n}\underset{i=1}{\overset{n}{\mathop \sum }}\,\left[ {{f}_{t}}\left( x_{i}^{t} \right)-{{f}_{t}}\left( x_{*}^{t} \right) \right] \right],
\end{align}
where $x_*^t\in arg \min_{x\in X}f_t\left(x\right), t=1,2,...T$. Due to $\frac{1}{n}\underset{i=1}{\overset{n}{\mathop \sum }}\,\left[ {{f}_{T}}\left( x_{i}^{T} \right)-{{f}_{T}}\left( x_{*}^{T} \right) \right] \le R_T^F$,  the tracking performance at $T$ can be given when the upper bound of $R_T^F$ is known. Clearly, if  ${{\lim }_{T\to \infty }}R_{T}^{F}=0$ or a non-zero small bounded quantity, the decisions generated by the algorithm are close enough as the optimal solution at $T$.
\par

\textit{Remark 1.}  For a distributed algorithm with sublinear dynamic regret, it cannot always guarantee good tracking performance at the final iteration. For example, we set ${{m}_{t}}=\frac{1}{n}\sum\limits_{i=1}^{n}{\left[ {{f}_{t}}\left( x_{i}^{t} \right)-{{f}_{t}}\left( x_{*}^{t} \right) \right]}=1$, if $t=3^m$ and $m_t=0$ otherwise. Then $\underset{T\to \infty }{\mathop{\lim }}\,\underset{t=1}{\overset{T}{\mathop \sum }}\,{{m}_{t}}/T\le \underset{T\to \infty }{\mathop{\lim }}\,\left( {{\log }_{3}}T \right)/T=0$, while there's no way to guarantee 
${{m}_{T}}=\frac{1}{n}\underset{i=1}{\mathop{\overset{n}{\mathop{\sum }}\,}}\,\left[ {{f}_{T}}\left( x_{i}^{T} \right)-{{f}_{T}}\left( x_{*}^{T} \right) \right]=0$.

Now, we make some assumptions:

\noindent\textbf{Assumption 1.} (\cite{a14}) Suppose $f_i^t\ (\cdot)$ is differentiable and convex in $X$. For all $t>0$, there exists $L>0$ such that $\left\| {\mathrm{\nabla}f_i^t(x)} \right\| \le L$ and $\{f_i^t{(\cdot)},t=1,...,T\}$ is $L_s-smooth$ in $X$, where $L_s > 0$ is a constant, i.t. for any $x,y \in X$, $f_{i}^{t}\left( y \right)-f_{i}^{t}\left( x \right)\le \left\langle \nabla f_{i}^{t}\left( x \right),y-x \right\rangle +\frac{{{L}_{s}}}{2}{{\left\| y-x \right\|}^{2}}.$

\noindent\textbf{Assumption 2.} The feasible set $X\subset {{\mathbb{R}}^{d}}$ is convex and compact. Define $M=\sup \{\parallel x\parallel : x\in X\}<\infty $.

\noindent\textbf{Assumption 3.} (\cite{a18}) The undirected topology graph meets the following three conditions for any $t\in {\mathbb{N}_{+}}$:

$(a)$ There exists $\omega  \in (0,1)$  such that ${\omega_{ij}}\ge \omega $ whenever ${\omega_{ij}}> 0$.

$(b)$ The adjacency matrix $W$ is doubly stochastic,  i.e.,

$\sum_{i=1}^{n}\omega_{ij}=\sum_{j=1}^{n}{\omega_{ij} =1}$, $\forall i,j\in [n]$.\par
$(c)$ The topology graph is connected at any time.\par
\textit{Remark 2.}
The above assumptions are important for the study of distributed algorithms. Assumption 1 can ensure the stability and convergence of the algorithm. The boundedness of decision in Assumption 2 is very important to ensure the feasibility of the optimization process and preventing the divergence of solutions. Assumption 3 has many benefits, such as simplifying the structure, ensuring algorithm convergence, balancing weights.  \par

Below, we give some lemmas, which would be used in the following analysis.
\begin{lemma} (\cite{a15, a16})
	If the adjacent  matrix $W=[w_{ij}]$ satisfies Assumption 3. Then\\
	\begin{equation}
		\left| {\left[ W^{t-s} \right]_{ij}}-\frac{1}{n} \right|\le \gamma {{\lambda }^{t-s}}, \forall i,j\in [n], \forall t\ge s\ge 1,
	\end{equation}
	where $\gamma ={{(1-{}^{\omega }/{}_{4{{n}^{2}}})}^{-2}}>1$, and $\lambda ={{(1-{}^{\omega }/{}_{4{{n}^{2}}})}}^{{1}/{}_B}\in (0,1)$.
	Under Assumption 3, $B$ can take any integer between $[1, T]$.\\
\end{lemma}

\begin{lemma} (\cite{a17})
	Let $K$ be a non-empty, closed, and convex subset within $\mathbb{R}^d$ and $m$, $n$ be two vectors in ${{\mathbb{R}}^{d}}$. If ${x}={P_{K}}(n-m)$, then
	\begin{equation}
		2\left\langle {x}-z,m \right\rangle \le {{\left\| z-\left. n \right\| \right.}^{2}}-\left\| z- \right.{{\left. {x} \right\|}^{2}}-\left\| {x}- \right.{{\left. n \right\|}^{2}}, \forall z\in K.
	\end{equation}
\end{lemma}

\begin{lemma} (\cite{a18})
	Suppose Assumption 3 holds. For all $i\in \left[ n \right]$ and $t\in {\mathbb{N}_{+}}$, denote $\varepsilon_{i,t-1}^{x}=x_{i}^{t}-z_{i}^{t}$, and $\overline{{{x}^{t}}}=\frac{1}{n}\sum\limits_{i=1}^{n}{x_{i}^{t}}$, where $x_{\text{i}}^{t}$ is the iteration sequence of agent $i$, then\\
	\begin{align}
		\left\| x_{i}^{t}-\overline{{{x}^{t}}}\, \right\|&\le \gamma {{\lambda }^{t-2}}\sum\limits_{j=1}^{n}{\left\| x_{j}^{1} \right\|}+\frac{1}{n}\sum\limits_{j=1}^{n}{\left\| \varepsilon _{j,t-1}^{x} \right\|}+\left\| \varepsilon _{i,t-1}^{x} \right\|\nonumber\\
		&+\gamma \sum\limits_{s=1}^{t-2}{{{\lambda }^{t-s-2}}}\sum\limits_{j=1}^{n}{\left\| \varepsilon _{j,s}^{x} \right\|}.
	\end{align}
\end{lemma}

\begin{lemma}(\cite{lemma4})
	Let $\{{{\gamma }_{k}}\}$ be a scalar sequence, if ${{\lim }_{k\to \infty }}{{\gamma }_{k}}=\gamma $ and $0 <{\beta}<1$, then
	\begin{equation}
		{{\lim }_{k\to \infty }}\sum\limits_{l=0}^{k}{{{\beta }^{k-l}}}{{\gamma }_{l}}=\frac{\gamma }{1-\beta}.
	\end{equation}
\end{lemma}

\section{DFFR of Distributed Online Gradient-free Algorithm}
Let $\mathbb{B}$ represent the unit ball and $\mathbb{S}$ denote the unit sphere (They are all $d$-dimensional). In the bandit-feedback problem, due to the absence of specific gradients $\nabla f_{i}^{t}\left( x_{i}^{t} \right)$, we propose the following method to estimate it. At time $t$, the local cost function $f_{i}^{t}\left( x_{i}^{t} \right)$ is estimated by the $\delta $-smoothing function 
\begin{equation}
	\overset{\wedge }{\mathop{f_{i,\delta }^{t}}}\,\left( x_{i}^{t} \right)\triangleq \mathbb{E} \left[ f_{i}^{t}\left( x_{i}^{t}+\delta v_i^t \right) \right],
\end{equation}
where $\delta > 0$ represents a constant, and $v_i^t$ is a random vector that is evenly distributed within the unit ball $\mathbb{B}$. Before introducing the algorithm, we propose a necessary assumption as follows.\par
\noindent\textbf{Assumption 4} (\cite{a14})
Let $r\mathbb{B}$ be a subset of $X$, which in turn is contained within $R\mathbb{B}$, where $R > r > 0$. Furthermore, there is a constant $L_1 > 0$ for which $\sup_{1 \leq t \leq T, x_i^t \in X} |f_i^t(x_i^t)| \leq L_1$.

\textit{Remark 3.}
$X$ is a set that is contained between an $R$-fold extension and an $r$-fold reduction of the unit ball $\mathbb{B}$. In other words, the norm of each point of $X$ is between $r$ and $R$, ensuring that $X$ is a set limited within this range.
\begin{lemma} (\cite{a20})
	Let $\delta \in \left( 0, r \right)$, and set ${{X}_{\delta }}\triangleq \left( 1-\frac{\delta }{r} \right)X$. The following conditions holds
	\begin{equation}
		{{X}_{\delta }}+\delta \mathbb{B}\subset X
	\end{equation}
	and
	\begin{equation}
		\nabla \overset{\wedge }{\mathop{f_{i,\delta }^{t}}}\,\left( x_{i}^{t} \right)=\frac{d}{\delta }\mathbb{E}\left[ f_{i}^{t}\left( x_{i}^{t}+\delta u_i^t \right)u_i^t \right],
	\end{equation}
	where  $u_i^t$ is a random vector that is homogeneously spread across the $d$-dimensional unit sphere $\mathbb{S}$.
\end{lemma}

\textit{Remark 4.} Due to the new loss function $\overset{\wedge }{\mathop{f_{i,\delta }^{t}}}\,$, its perturbations may move points outside the feasible set. To deal with this situation, we limit its feasible set to ${{X}_{\delta }}$, which can satisfy the ball of radius $\delta$ around each point in the subset is  contained in $X$.\par
Now, we propose the algorithm as follows.
\hspace*{\fill} \\
{\noindent}	 \rule[0pt]{9cm}{0.1em}\\
\textbf{Algorithm 1}: Distributed Online Gradient-free Algorithm\\
{\noindent}	 \rule[5pt]{9cm}{0.05em}\\
\textbf{Input}: positive and non-increasing step size sequence $\{{\alpha}_t\}$, the final moment $T$ of iteration, a smoothing parameter $\delta$.\\
\textbf{Initialize}: $x_{i}^{1}\in {{X}_{\delta }}$ for all $i\in \left[ n \right]$.\\
\textbf{For} $t = 1,...,T$ \par choose a random vector ${{u}_i^{t}}$, which is evenly distributed over unit sphere $\mathbb{S}$ and independent of $\left\{ {{u}_i^{1}},...,{{u}_i^{T-1}} \right\}$\par
\textbf{for} $i = 1,...,n$ \par update the decision as
\begin{align}
	&g_{i}^{t}=\frac{d}{\delta }\left( f_{i}^{t}(x_{i}^{t}+\delta {{u}_i^{t}})-f_{i}^{t}(x_{i}^{t}) \right){{u}_i^{t}},\\
	&z_i^{t+1}=\sum_{j=1}^{n}{w_{ij}x_j^t},\\
	&{\hat{x}}_i^{t+1}=z_i^{t+1}-{\alpha}_tg_i^t,\\
	&x_{i}^{t+1}={{P}_{{{X}_{\delta }}}}({\hat{x}}_i^{t+1}),
\end{align}\par
\textbf{end for}\\
\textbf{end for}\\
\textbf{Output: $\{x_i^t\}$}.\\
{\noindent}	 \rule[5pt]{9cm}{0.05em}

\begin{theorem}
	Suppose Assumptions 1-4 hold. Let ${{\{{{x}_{i}^{t}}\}}_ {t=1}^{T}}$ be the decision sequence generated by Algorithm 1 with $\left\{ {{\alpha }_{t}} \right\}_{t=1}^{T}$, which is a positive decreasing sequence. Denote ${\widetilde{x_{*}^{t}}}\,=\arg {{\min }_{x_i^t\in {{X}_{\delta }}}}{{f}_{t}}(x)$, ${{F}_{i}^{t}}\triangleq \left| \left\| \varepsilon _{i,t}^{x} \right\|-\left\| \varepsilon _{i,t-1}^{x} \right\| \right|$, ${\widetilde{{\theta}_{t}}}\,\triangleq \left\| {\widetilde{x_{*}^{t}}}\,-{\widetilde{x_{*}^{t+1}}}\, \right\|$, $\sigma \triangleq (4+5d)L+\frac{{2L(1+d){\rho }^{-2}}}{1-\frac{\lambda }{\rho }}$,  then for any $1 > \rho >\lambda > 0$,
	\begin{align}
		&\mathbb{E}\left[ R_{T}^{F} \right]\le 2L(1+d)\gamma \underset{j=1}{\overset{n}{\mathop \sum }}\,\mathbb{E}\parallel x_{j}^{1}\parallel \underset{t=1}{\overset{T}{\mathop \sum }}\,{{\rho }^{T-t}}{{\lambda }^{t-2}}\nonumber
	\end{align}
	\begin{align}
		&+\frac{4L(1+d)}{n}\underset{t=1}{\overset{T}{\mathop \sum }}\,{{\rho }^{T-t}}\sum\limits_{i=1}^{n}\mathbb{E}{\left({{F}_{i}^{t}}\right)}+\sum\limits_{t=1}^{T}{\frac{{\alpha}_t n}{2}}{{\sigma }^{2}}{{\rho }^{T-t}}\nonumber\\
		&+\sum\limits_{t=1}^{T}{\frac{\delta }{r}{{L}_{1}}}{{\rho }^{T-t}}+2M\sum\limits_{t=1}^{T}{{{\rho }^{T-t}}}\frac{\mathbb{E}{\left({\widetilde{{\theta}_{t}}}\right)}}{{{\alpha}}_t}+\frac{2{{M}^{2}}}{{\alpha}_T}.
	\end{align}
\end{theorem}

\begin{proof}
	By the definition of ${R}_T^F$, we start with $\frac{1}{n}\sum\limits_{i=1}^{n}{{{f}_{t}}}\left( x_{i}^{t} \right)$,
	\vspace{-\abovedisplayskip}
	\begin{align}
		&\frac{1}{n}\sum\limits_{i=1}^{n}{{{f}_{t}}}\left( x_{i}^{t} \right)=\frac{1}{n}\sum\limits_{i=1}^{n}{\left[ \frac{1}{n}\sum\limits_{j=1}^{n}{f_{j}^{t}}\left( x_{i}^{t} \right) \right]}\nonumber\\
		&=\frac{1}{n}\sum\limits_{i=1}^{n}{f_{i}^{t}}\left( x_{i}^{t} \right)+\frac{1}{{{n}^{2}}}\sum\limits_{i=1}^{n}{\sum\limits_{j=1}^{n}{(f_{j}^{t}}}(x_{i}^{t})-f_{j}^{t}(x_{j}^{t}))\nonumber\\
		&\le \frac{1}{n}\sum\limits_{i=1}^{n}{f_{i}^{t}}\left( x_{i}^{t} \right)+\frac{1}{{{n}^{2}}}\sum\limits_{i=1}^{n}{\sum\limits_{j=1}^{n}{L\left\| x_{i}^{t}-x_{j}^{t} \right\|}}
		\nonumber\\
		&\le \frac{1}{n}\sum\limits_{i=1}^{n}{f_{i}^{t}}\left( x_{i}^{t} \right)+\frac{2L}{n}\sum\limits_{i=1}^{n}{\left\| x_{i}^{t}-\overline{{{x}^{t}}}\, \right\|}.
	\end{align}
	Thus, we can get an upper bound of the latter part of the regret as defined earlier in $(2)$.
	\begin{align}
		&\frac{1}{n}\underset{i=1}{\overset{n}{\mathop \sum }}\,\left[ {{f}_{t}}\left( x_{i}^{t} \right)-{{f}_{t}}\left( x_{*}^{t} \right) \right] \le \frac{1}{n}\sum\limits_{i=1}^{n}{f_{i}^{t}}\left( x_{i}^{t} \right)-\frac{1}{n}\sum\limits_{i=1}^{n}{f_{i}^{t}}\left( x_{*}^{t} \right)\nonumber\\
		&+\frac{2L}{n}\sum\limits_{i=1}^{n}{\left\| x_{i}^{t}-\overline{{{x}^{t}}}\, \right\|}\nonumber\\
		&= \frac{1}{n}\underset{i=1}{\overset{n}{\mathop \sum }}\,\left[ f_{i}^{t}\left( x_{i}^{t} \right)-f_{i}^{t}\left( x_{*}^{t} \right) \right]+\frac{2L}{n}\sum\limits_{i=1}^{n}{\left\| x_{i}^{t}-\overline{{{x}^{t}}}\, \right\|}.
	\end{align}
	Notice ${\widetilde{x_{*}^{t}}}\,=\arg {{\min }_{x_i^t\in {{X}_{\delta }}}}{{f}_{t}}(x)$ and we know that $0\in X$. Then by the convexity of $f_{i}^{t}(\cdot )$, we have
	\begin{align}
		{{f}_{t}}(\widetilde{x_{*}^{t}})
		&={{\min }_{x_{i}^{t}\in {{X}_{\delta }}}}{{f}_{t}}(x_{i}^{t})
		={{\min }_{x_{i}^{t}\in X}}{{f}_{t}}[(1-\frac{\delta }{r})x_{i}^{t}]	\nonumber\\
		&={{\min }_{x_{i}^{t}\in X}}{{f}_{t}}[0\cdot \frac{\delta }{r}+(1-\frac{\delta }{r})x_{i}^{t}]\nonumber\\
		&\le {{\min }_{x_{i}^{t}\in X}}\left\{ \frac{\delta }{r}{{f}_{t}}(0)+(1-\frac{\delta }{r}){{f}_{t}}(x_{i}^{t}) \right\}.
	\end{align}
	Due to Assumption 4 and ${{\sup }_{1\le t\le T,x_i^t\in X}}\left| f_{i}^{t}\left( x_{i}^{t} \right) \right|\le {{L}_{1}}$, we get
	\begin{equation}
		{{f}_t}(\widetilde{x_{*}^{t}})\le \frac{\delta }{r}{{L}_{1}}+f_{t}(x_{*}^{t}).
	\end{equation}
	Therefore, we can get 
	\begin{equation}
		\sum\limits_{\text{i}=1}^{n}{\left[ f_{i}^{t}\left(\widetilde{x_{*}^{t}} \right)-f_{i}^{t}\left( x_{*}^{t} \right) \right]}\le \frac{n\delta }{r}{{L}_{1}}.
	\end{equation}	
	By Assumption 1, for any $x_{i}^{t}\in {{X}_{\delta }}$, we have the following estimate
	\begin{align}
		\left| \overset{\wedge }{\mathop{f_{i,\delta }^{t}}}\,(x_{i}^{t})-f_{i}^{t}(x_{i}^{t}) \right|
		&=\left|\mathbb{E}\left[ {{f}_{t}}(x_{i}^{t}+\delta v_i^t) \right]-{{f}_{t}}(x_{i}^{t}) \right|\nonumber\\
		&\le \mathbb{E}\left| f_{i}^{t}(x_{i}^{t}+\delta v_i^t)-f_{i}^{t}(x_{i}^{t}) \right|.
	\end{align}	
	Based on the convexity of $f_{i}^{t}$, we obtain
	\begin{align}
		\left| \overset{\wedge }{\mathop{f_{i,\delta }^{t}}}\,(x_{i}^{t})-f_{i}^{t}(x_{i}^{t}) \right|&\le \mathbb{E}\left[ L\delta \left\| v_i^t \right\| \right]\le L\delta.
	\end{align}
	Now we can convert the previous part $f_{i}^{t}(x_{i}^{t})-f_{i}^{t}(x_{*}^{t})$ as follows
	\begin{align}
		&\sum\limits_{i=1}^{n}{\left( f_{i}^{t}(x_{i}^{t})-f_{i}^{t}(x_{*}^{t}) \right)}\le \sum\limits_{i=1}^{n}{\left( \overset{\wedge }{\mathop{f_{i,\delta }^{t}}}\,\left( x_{i}^{t} \right)+L\delta  \right)}\nonumber\\
		&-\sum\limits_{i=1}^{n}{\left(\overset{\wedge }{f_{i,\delta }^{t}} ({\widetilde{x_{*}^{t}}}\,)+L\delta  \right)}+\frac{n\delta }{r}{{L}_{1}}\nonumber\\
		&=\sum\limits_{i=1}^{n}{\left( \overset{\wedge }{\mathop{f_{i,\delta }^{t}}}\,\left( x_{i}^{t} \right)-\overset{\wedge }{f_{i,\delta }^{t}}(\widetilde{x_{*}^{t}}) \right)}+\frac{n\delta }{r}{{L}_{1}}.
	\end{align}	
	Then, let's analyze of the upper bound of $\overset{\wedge }{\mathop{f_{i,\delta }^{t}}}\,\left( x_{i}^{t} \right)-\overset{\wedge }{f_{i,\delta }^{t}}(\widetilde{x_{*}^{t}})$. Note that
	\begin{equation}
		\overset{\wedge }{\mathop{f_{i,\delta }^{t}}}\,\left( x_{i}^{t} \right)-\overset{\wedge }{f_{i,\delta }^{t}}({\widetilde{x_{*}^{t}}}\,)\le \left\langle \nabla \overset{\wedge }{\mathop{f_{i,\delta }^{t}}}\,(x_{i}^{t}),x_{i}^{t}-{\widetilde{x_{*}^{t}}}\, \right\rangle
	\end{equation} 	
	and the conditional expectation $\mathbb{E}[g_{i}^{t}(x_{i}^{t})|x_{i}^{t}]=\nabla \overset{\wedge }{\mathop{f_{i,\delta }^{t}}}\,(x_{i}^{t})$, we have
	\begin{align}
		\mathbb{E}\left\langle g_{i}^{t},x_{i}^{t}-\widetilde{x_{*}^{t}} \right\rangle =\left\langle \nabla \overset{\wedge }{\mathop{f_{i,\delta }^{t}}}\,(x_{i}^{t}),x_{i}^{t}-{\widetilde{x_{*}^{t}}}\, \right\rangle,
	\end{align}
	from which and by taking the mathematical expectation to both sides of $(23)$, we have
	\begin{align}
		\mathbb{E}\left({\overset{\wedge }{\mathop{f_{i,\delta }^{t}}}\,\left( x_{i}^{t} \right)-\overset{\wedge }{f_{i,\delta }^{t}}({\widetilde{x_{*}^{t}}}\,)}\right) \le \mathbb{E}\left\langle g_{i}^{t},x_{i}^{t}-\widetilde{x_{*}^{t}} \right\rangle.
	\end{align}	
	We introduce $x_{i}^{t+1}$ in the follow equation:
	\begin{align}
		\mathbb{E}\left\langle g_{i}^{t},x_{i}^{t}-\widetilde{x_{*}^{t}} \right\rangle=\mathbb{E}\left\langle g_{i}^{t},x_{i}^{t}-x_{i}^{t+1} \right\rangle +\mathbb{E}\left\langle g_{i}^{t},x_{i}^{t+1}-\widetilde{x_{*}^{t}} \right\rangle .
	\end{align} 		
	By Assumption 1, $\left\| g_{i}^{t} \right\|\le \frac{d}{\delta }L\left\| \delta u_{i}^{t} \right\|=dL$, we can get
	\begin{equation}
		\mathbb{E}\left\langle g_{i}^{t},x_{i}^{t}-\widetilde{x_{*}^{t}} \right\rangle\le dL\mathbb{E}\left\| x_{i}^{t}-x_{i}^{t+1} \right\|+\mathbb{E}\left\langle g_i^t,x_{i}^{t+1}-{\widetilde{x_{*}^{t}}}\, \right\rangle.
	\end{equation} 		
	By Lemma 2, we get
	\begin{equation}
		\begin{aligned}
			\mathbb{E}\left\langle g_{i}^{t},x_{i}^{t+1}-{\widetilde{x_{*}^{t}}} \right\rangle\le \frac{1}{2{\alpha}_t}\mathbb{E}\left[ \left\| {\widetilde{x_{*}^{t}}} \right.-
			{{\left. z_{i}^{t+1} \right\|}^{2}}-\left\| {\widetilde{x_{*}^{t+1}}} \right.-{{\left. z_{i}^{t+2} \right\|}^{2}} 
			\right. \\  \left.-{{\left\| \varepsilon _{i,t}^{x} \right\|}^{2}}+\left\| {\widetilde{x_{*}^{t+1}}} \right.-{{\left. x_{i}^{t+1} \right\|}^{2}}-\left\| {\widetilde{x_{*}^{t}}} \right.-{{\left. x_{i}^{t+1} \right\|}^{2}}\text{ } \right].
		\end{aligned}
	\end{equation}	
	Here,
	\begin{align}
		&\mathbb{E}\left[ R_{T}^{F} \right]\le \frac{dL}{n}\sum\limits_{t=1}^{T}{{{\rho }^{T-t}}}\sum\limits_{i=1}^{n}\mathbb{E}{\left\| x_{i}^{t}-x_{i}^{t+1} \right\|}\nonumber\\&+\frac{2L}{n}\underset{t=1}{\overset{T}{\mathop \sum }}\,{{\rho }^{T-t}}\underset{i=1}{\overset{n}{\mathop \sum }}\,\mathbb{E}\parallel x_{i}^{t}-\overline{{{x}^{t}}}\parallel \nonumber\\
		&+\frac{1}{2n}\sum\limits_{t=1}^{T}{{{\rho }^{T-t}}}\frac{1}{{\alpha}_t}\sum\limits_{i=1}^{n}\mathbb{E}{\left( {{\left\| {\widetilde{x_{*}^{t}}}\,-z_{i}^{t+1} \right\|}^{2}}-{{\left\| {\widetilde{x_{*}^{t+1}}}\,-z_{i}^{t+2} \right\|}^{2}} \right)}\nonumber\\
		&+\frac{1}{2n}\sum\limits_{t=1}^{T}{{{\rho }^{T-t}}}\frac{1}{{\alpha}_t}\sum\limits_{i=1}^{n}\mathbb{E}{({{\left\| {\widetilde{x_{*}^{t+1}}}\,-x_{i}^{t+1} \right\|}^{2}}-{{\left\| {\widetilde{x_{*}^{t}}}\,-x_{i}^{t+1} \right\|}^{2}})}\nonumber\\
		&-\frac{1}{2n}\sum\limits_{t=1}^{T}{{{\rho }^{T-t}}\frac{1}{{\alpha}_t}\sum\limits_{i=1}^{n}\mathbb{E}{{{\left\| \varepsilon _{i,t}^{x} \right\|}^{2}}}}+\underset{t=1}{\overset{T}{\mathop \sum }}\,{{\rho }^{T-t}}\frac{\delta {{L}_{1}}}{r}.
	\end{align}	
	Now, let's analyze the upper bound of the first part in $(29)$
	\begin{align}
		\mathbb{E}\left\| x_{i}^{t} \right.-\left. x_{i}^{t+1} \right\|&=\mathbb{E}\left\| x_{i}^{t}-z_{i}^{t+1}+z_{i}^{t+1}-x_{i}^{t+1} \right\|\nonumber\\
		&\le \mathbb{E}\left\| x_{i}^{t}-z_{i}^{t+1} \right\|+\mathbb{E}\left\| \varepsilon _{i,t}^{x} \right\| \nonumber \\
		&\le \mathbb{E}\left\| x_{i}^{t}-\overline{{{x}^{t}}}\,  \right\|+\mathbb{E}\left\|\overline{{{x}^{t}}}\, -z_{i}^{t+1} \right\|+\mathbb{E}\left\| \varepsilon _{i,t}^{x} \right\|.
	\end{align}	
	Then, through (11), we can obtain
	\begin{align}
		&\sum\limits_{i=1}^{n}\mathbb{E}{\left\| x_{i}^{t} \right.-\left. x_{i}^{t+1} \right\|}
		\le \sum\limits_{i=1}^{n}\mathbb{E}{\left\| x_{i}^{t}-\overline{{{x}^{t}}}\,  \right\|}\nonumber\\
		&+\sum\limits_{i=1}^{n}\mathbb{E}{\left\| \overline{{{x}^{t}}}\, -\sum\limits_{j=1}^{n}{{{w}_{ij}}x_{j}^{t}} \right\|}
		+\sum\limits_{i=1}^{n}\mathbb{E}{\left\| \varepsilon _{i,t}^{x} \right\|}.
	\end{align}	
	Note that $\left\| \cdot  \right\|$ is convex and Assumption 3 $(b)$, we can get
	\begin{align}
		&\sum\limits_{i=1}^{n}\mathbb{E}{\left\| x_{i}^{t} \right.-\left. x_{i}^{t+1} \right\|}\le \sum\limits_{i=1}^{n}\mathbb{E}{\left\| x_{i}^{t}-\overline{{{x}^{t}}}\,  \right\|}\nonumber\\&+\sum\limits_{i=1}^{n}{\sum\limits_{j=1}^{n}{{{w}_{ij}}}\mathbb{E}\left\| \overline{{{x}^{t}}}\, -x_{j}^{t} \right\|}+\sum\limits_{i=1}^{n}\mathbb{E}{\left\| \varepsilon _{i,t}^{x} \right\|}\nonumber \\
		&=2\sum\limits_{i=1}^{n}\mathbb{E}{\left\| x_{i}^{t}-\overline{{{x}^{t}}}\,  \right\|}+\sum\limits_{i=1}^{n}\mathbb{E}{\left\| \varepsilon _{i,t}^{x} \right\|}.
	\end{align}	
	Then, by Lemma 3, we can obtain
	\begin{align}
		&\sum\limits_{t=1}^{T}{{{\rho }^{T-t}}}\sum\limits_{i=1}^{n}{\gamma {{\lambda }^{t-2}}}\sum\limits_{j=1}^{n}\mathbb{E}{\left\| x_{j}^{1} \right\|}\nonumber\\&=\sum\limits_{i=1}^{n}{\gamma }\sum\limits_{j=1}^{n}\mathbb{E}{\left\| x_{j}^{1} \right\|}\sum\limits_{t=1}^{T}{{{\rho }^{T-t}}}{{\lambda }^{t-2}}.
	\end{align}	
	For the second term of $x_{i}^{t}-\overline{{x}^{t}}$,
	\begin{align}
		\sum\limits_{t=1}^{T}{{{\rho }^{T-t}}}\sum\limits_{i=1}^{n}{\frac{1}{n}\sum\limits_{j=1}^{n}\mathbb{E}{\left\| \varepsilon _{j,t-1}^{x} \right\|}}&=\sum\limits_{t=1}^{T}{{{\rho }^{T-t}}}\sum\limits_{j=1}^{n}\mathbb{E}{\left\| \varepsilon _{j,t-1}^{x} \right\|}\nonumber\\
		&=\sum\limits_{t=1}^{T}{{{\rho }^{T-t}}}\sum\limits_{i=1}^{n}\mathbb{E}{\left\| \varepsilon _{i,t-1}^{x} \right\|}.
	\end{align}	
	Note that ${{F}_{i}^{t}}\triangleq \left| \left\| \varepsilon _{i,t}^{x} \right\|-\left\| \varepsilon _{i,t-1}^{x} \right\| \right|$, we have
	\begin{align}
		&\sum\limits_{t=1}^{T}{{{\rho }^{T-t}}}\sum\limits_{i=1}^{n}\mathbb{E}{\left\| \varepsilon _{i,t-1}^{x} \right\|}\le \sum\limits_{t=1}^{T}{{{\rho }^{T-t}}}\sum\limits_{i=1}^{n}\mathbb{E}{\left( \left\| \varepsilon _{i,t}^{x} \right\|+{{F}_{i}^{t}} \right)}\nonumber\\
		&=\sum\limits_{t=1}^{T}{{{\rho }^{T-t}}}\sum\limits_{i=1}^{n}\mathbb{E}{\left\| \varepsilon _{i,t}^{x} \right\|+}\sum\limits_{t=1}^{T}{{{\rho }^{T-t}}}\sum\limits_{i=1}^{n}\mathbb{E}{\left({{F}_{i}^{t}}\right)}.\
	\end{align}	
	Now, for the last term of (5), 
	\begin{align}
		&\sum\limits_{t=1}^{T}{{{\rho }^{T-t}}}\sum\limits_{s=1}^{t-2}{{{\lambda }^{t-s-2}}}\sum\limits_{j=1}^{n}\mathbb{E}{\left\| \varepsilon _{j,s}^{x} \right\|}\nonumber\\&=\sum\limits_{s=1}^{T-2}{\sum\limits_{t=s+2}^{T}{{{\rho }^{T-t}}}}{{\lambda }^{t-s-2}}\sum\limits_{j=1}^{n}\mathbb{E}{\left\| \varepsilon _{j,s}^{x} \right\|}\nonumber\\
		&=\sum\limits_{s=1}^{T-2}{{{\rho }^{T-t-2}}}\sum\limits_{t=s+2}^{T}{{{\rho }^{-\left( t-s-2 \right)}}}{{\lambda }^{t-s-2}}\sum\limits_{j=1}^{n}\mathbb{E}{\left\| \varepsilon _{j,s}^{x} \right\|}\nonumber
	\end{align}	
	\begin{align}
		&=\sum\limits_{s=1}^{T-2}{{{\rho }^{T-t-2}}}\sum\limits_{v=0}^{T-s-2}{{{\left( \frac{\lambda }{\rho } \right)}^{v}}}\sum\limits_{j=1}^{n}\mathbb{E}{\left\| \varepsilon _{j,s}^{x} \right\|}.
	\end{align}	
	Owing to $\lambda <\rho $, we get
	\begin{align}
		&\sum\limits_{t=1}^{T}{{{\rho }^{T-t}}}\sum\limits_{s=1}^{t-2}{{{\lambda }^{t-s-2}}}\sum\limits_{j=1}^{n}\mathbb{E}{\left\| \varepsilon _{j,s}^{x} \right\|}\nonumber\\&\le \frac{1}{1-\frac{\lambda }{\rho }}\sum\limits_{s=1}^{T-2}{{{\rho }^{T-s-2}}}\sum\limits_{j=1}^{n}\mathbb{E}{\left\| \varepsilon _{j,s}^{x} \right\|}\nonumber\\
		&=\frac{{{\rho }^{-2}}}{1-\frac{\lambda }{\rho }}\sum\limits_{t=1}^{T-2}{{{\rho }^{T-t}}}\sum\limits_{j=1}^{n}\mathbb{E}{\left\| \varepsilon _{j,t}^{x} \right\|}\nonumber\\
		&\le \frac{{{\rho }^{-2}}}{1-\frac{\lambda }{\rho }}\sum\limits_{\text{t}=1}^{T}{{{\rho }^{T-t}}}\sum\limits_{i=1}^{n}\mathbb{E}{\left\| \varepsilon _{i,t}^{x} \right\|}.
	\end{align}
	Therefore, through $(33)-(37)$, we can analyze that the DFFR in the first two parts of $(29)$ is as follows
	\begin{align}
		&\frac{dL}{n}\sum\limits_{t=1}^{T}{{{\rho }^{T-t}}}\sum\limits_{i=1}^{n}\mathbb{E}{\left\| x_{i}^{t}-x_{i}^{t+1} \right\|}\nonumber\\
		&+\frac{2L}{n}\underset{t=1}{\mathop{\overset{T}{\mathop{\sum }}\,}}\,{{\rho }^{T-t}}\underset{i=1}{\mathop{\overset{n}{\mathop{\sum }}\,}}\,\mathbb{E}\left\| x_{i}^{t}-\overline{{{x}^{t}}} \right\|\nonumber\\
		&\le 2L(1+d)\gamma \underset{j=1}{\mathop{\overset{n}{\mathop{\sum }}\,}}\,\mathbb{E}\parallel x_{j}^{1}\parallel \underset{t=1}{\mathop{\overset{T}{\mathop{\sum }}\,}}\,{{\rho }^{T-t}}{{\lambda }^{t-2}}\nonumber\\&+\frac{4L(1+d)}{n}\underset{t=1}{\mathop{\overset{T}{\mathop{\sum }}\,}}\,{{\rho }^{T-t}}\sum\limits_{i=1}^{n}\mathbb{E}{\left({F}_{i}^{t}\right)}\nonumber\\
		&+(\frac{(4+5d)L}{n}+\frac{\frac{(2L(1+d))}{n}{{\rho }^{-2}}}{1-{}^{\lambda }/{}_{\rho }}\underset{t=1}{\overset{T}{\mathop )\sum }}\,{{\rho }^{T-t}}\sum\limits_{i=1}^{n}\mathbb{E}{\left\| \varepsilon _{i,t}^{x} \right\|}.
	\end{align}	
	Next, we will analyze the upper bound of the DFFR in third  part of $(29)$ 
	\begin{equation}
		\sum\limits_{t=1}^{T}{{{\rho }^{T-t}}}\frac{1}{2{n{a}_{t}}}\sum\limits_{i=1}^{n}\mathbb{E}{\left( {{\left\| {\widetilde{x_{*}^{t}}}\,-z_{i}^{t+1} \right\|}^{2}}-{{\left\| {\widetilde{x_{*}^{t+1}}}\,-z_{i}^{t+2} \right\|}^{2}} \right)}\nonumber
	\end{equation}
	\begin{equation}
		\begin{aligned}
			\le \frac{1}{n}\sum\limits_{i=1}^{n}{\left( \frac{{{\rho }^{T-1}}}{2{{a}_{1}}}\mathbb{E}{{\left\| {\widetilde{x_{*}^{1}}}\,-z_{i}^{2} \right\|}^{2}}-\frac{{{\rho }^{0}}}{2{{a}_{T}}}\mathbb{E}{{\left\| {\widetilde{x_{*}^{T+1}}}\,-z_{i}^{T+2} \right\|}^{2}}\right. }\\{\left.+\sum\limits_{t=0}^{T-2}{\left( \frac{{{\rho }^{t}}}{2{{a}_{T-t}}}-\frac{{{\rho }^{t+1}}}{2{{a}_{T-t-1}}} \right)}4{{M}^{2}} \right)}	\nonumber\\
		\end{aligned}
	\end{equation}
	\begin{align}
		&\le \frac{1}{n}\sum\limits_{i=1}^{n}{\left( \frac{{{\rho }^{T-1}}}{2{{a}_{1}}}\mathbb{E}{{\left\| {\widetilde{x_{*}^{1}}}\,-z_{i}^{2} \right\|}^{2}}+\left( \frac{{{\rho }^{0}}}{2{{a}_{T}}}-\frac{{{\rho }^{T-1}}}{2{{a}_{1}}} \right)4{{M}^{2}} \right)}\nonumber\\
		&\le \frac{2{M}^{2}}{{{a}_{T}}}.
	\end{align}
	Next, we will analyze the DFFR in the fourth  part of $(29)$. Noting that ${\widetilde{{\theta}_{t}}}\,\triangleq \left\| {\widetilde{x_{*}^{t}}}\,-{\widetilde{x_{*}^{t+1}}}\, \right\|$,
	\begin{align}
		&\mathbb{E}{{\left\| {\widetilde{x_{*}^{t+1}}}\,-x_{i}^{t+1} \right\|}^{2}}-\mathbb{E}{{\left\| {\widetilde{x_{*}^{t}}}\,-x_{i}^{t+1} \right\|}^{2}}\nonumber\\
		&\le \mathbb{E}\left\| {\widetilde{x_{*}^{t+1}}}\,-{\widetilde{x_{*}^{t}}}\, \right\|\cdot \mathbb{E} \left\| {\widetilde{x_{*}^{t+1}}}\,-x_{i}^{t+1}+{\widetilde{x_{*}^{t}}}\,-x_{i}^{t+1} \right\|\nonumber\\&\le 4M\mathbb{E}{\left({\widetilde{{\theta}_{t}}}\,\right)}.
	\end{align}	
	Therefore, we can obtain
	\begin{align}
		&\frac{1}{2{\alpha}_tn}\sum\limits_{t=1}^{T}{{{\rho }^{T-t}}}\sum\limits_{i=1}^{n}\mathbb{E}{({{\left\|{\widetilde{x_{*}^{t+1}}}\,-x_{i}^{t+1} \right\|}^{2}}-{{\left\| {\widetilde{x_{*}^{t}}}\,-x_{i}^{t+1} \right\|}^{2}})}\nonumber\\&\le 2M\sum\limits_{t=1}^{T}{{{\rho }^{T-t}}}\frac{\mathbb{E}{\left({\widetilde{{\theta}_{t}}}\right)}}{{{\alpha}}_t}.
	\end{align}
	Now, let's analyze the DFFR in the fifth part of $(29)$.
	Denote $\sigma \triangleq (4+5d)L+\frac{{2L(1+d){\rho }^{-2}}}{1-\frac{\lambda }{\rho }}$, we have
	\begin{equation}
		\sigma \mathbb{E}\left\| \varepsilon _{i,t}^{x} \right\|\le \frac{\mathbb{E}{{\left\| \varepsilon _{i,t}^{x} \right\|}^{2}}}{2{\alpha}_tn}+\frac{{\alpha}_tn}{2}{{\sigma }^{2}}.
	\end{equation}	
	Therefore, we can obtain
	\begin{align}
		&\underset{t=1}{\mathop{\overset{T}{\mathop{\sum }}\,}}\,{{\rho }^{T-t}}\underset{i=1}{\overset{n}{\mathop \sum }}\,\sigma \mathbb{E}\parallel \varepsilon _{i,t}^{x}\parallel -\frac{1}{2{{\alpha }_{t}}n}\underset{t=1}{\overset{T}{\mathop \sum }}\,{{\rho }^{T-t}}\underset{i=1}{\overset{n}{\mathop \sum }}\,\mathbb{E}\parallel \varepsilon _{i,t}^{x}{{\parallel }^{2}}\nonumber\\&\le \sum\limits_{t=1}^{T}{\frac{{\alpha}_t n}{2}{{\sigma }^{2}}{{\rho }^{T-t}}}.
	\end{align}
	From $(38), (39), (41)$ and $(43)$, we can get $(14)$. This proof is completed.\par
\end{proof}

\textit{Remark 5.} Theorem 1 gives the upper bound of the DFFR of Distributed Online  Gradient-free Algorithm. Note that the upper bound of $R_T^F$ is related to ${\delta}$, because we use $(9)$ as an estimate for $\nabla \overset{\wedge }{\mathop{f_{i,\delta }^{t}}}\,\left( x_{i}^{t} \right)$ , instead of the gradient of the local optimization function $\mathrm{\nabla}f_i^t\left(x_i^t\right)$ itself. \par
\noindent\textbf{Corollary 1.} Setting the step size ${\alpha}$ is a positive constant, we will further draw the following conclusion.
\begin{align}
	\mathbb{E}\left[ R_{T}^{F} \right]&\le 2L(1+d)\gamma 	\underset{j=1}{\overset{n}{\mathop \sum }}\,\mathbb{E}\parallel x_{j}^{1}\parallel \underset{t=1}{\overset{T}{\mathop \sum }}\,{{\rho }^{T-t}}{{\lambda }^{t-2}}\nonumber\\
	&+\frac{4L(1+d)}{n}\underset{t=1}{\overset{T}{\mathop \sum }}\,{{\rho }^{T-t}}\sum\limits_{i=1}^{n}\mathbb{E}{\left({{F}_{i}^{t}}\right)}+{\frac{{\sigma }^{2}{\alpha}n}{2}}\sum\limits_{t=1}^{T}{{\rho }^{T-t}}\nonumber\\
	&+{\frac{\delta }{r}{{L}_{1}}}\sum\limits_{t=1}^{T}{{\rho }^{T-t}}+\frac{2M}{\alpha}\sum\limits_{t=1}^{T}{{{\rho }^{T-t}}}\mathbb{E}{\left({\widetilde{{\theta}_{t}}}\right)}+\frac{2{{M}^{2}}}{{\alpha}}.
\end{align}

\begin{proof} Consider the upper bound of the DFFR in third  part of $(29)$, 
	\begin{align}
		&\sum\limits_{t=1}^{T}{{{\rho }^{T-t}}}\frac{1}{2{n{\alpha}}}\sum\limits_{i=1}^{n}\mathbb{E}{\left( {{\left\| {\widetilde{x_{*}^{t}}}\,-z_{i}^{t+1} \right\|}^{2}}-{{\left\| {\widetilde{x_{*}^{t+1}}}\,-z_{i}^{t+2} \right\|}^{2}} \right)}\nonumber
	\end{align}
		\begin{equation}
		\begin{aligned}
		\le \frac{1}{2\alpha n}\sum\limits_{i=1}^{n}{\left( {{\rho }^{T-1}}\mathbb{E}{{\left\| {\widetilde{x_{*}^{1}}}\,-z_{i}^{2} \right\|}^{2}}-\mathbb{E}{{\left\| {\widetilde{x_{*}^{T+1}}}\,-z_{i}^{T+2} \right\|}^{2}}\right. }\\{\left.+\sum\limits_{t=0}^{T-2}{\left( {{\rho }^{t}}-{{\rho }^{t+1}}\right)}4{{M}^{2}} \right)}	\nonumber\\
		\end{aligned}
		\end{equation}
	\begin{align}
		&\le \frac{1}{2\alpha n}\sum\limits_{i=1}^{n}{\left({{\rho }^{T-1}}\mathbb{E}{{\left\| {\widetilde{x_{*}^{1}}}\,-z_{i}^{2} \right\|}^{2}}+\left( {{\rho }^{0}}-{{\rho }^{T-1}} \right)4{{M}^{2}} \right)}\nonumber\\
		&\le \frac{2{M}^{2}}{\alpha}.
	\end{align}
	For the DFFR in the fourth  part of $(29)$, we have
	\begin{align}
		&\frac{1}{2{\alpha}n}\sum\limits_{t=1}^{T}{{{\rho }^{T-t}}}\sum\limits_{i=1}^{n}\mathbb{E}{({{\left\|{\widetilde{x_{*}^{t+1}}}\,-x_{i}^{t+1} \right\|}^{2}}-{{\left\| {\widetilde{x_{*}^{t}}}\,-x_{i}^{t+1} \right\|}^{2}})}\nonumber\\
		&\le \frac{2M}{\alpha}\sum\limits_{t=1}^{T}{{{\rho }^{T-t}}}\mathbb{E}{\left({\widetilde{{\theta}_{t}}}\right)}.
	\end{align}
	For the last part of DFFR, we use the following inequality:
	\begin{equation}
		\sigma \mathbb{E}\left\| \varepsilon _{i,t}^{x} \right\|\le \frac{\mathbb{E}{{\left\| \varepsilon _{i,t}^{x} \right\|}^{2}}}{2{\alpha} n}+\frac{{\alpha}n}{2}{{\sigma }^{2}}.
	\end{equation}
	Then, we have 
	\begin{align}
		&\sum\limits_{t=1}^{T}{{{\rho }^{T-t}}}\underset{i=1}{\overset{n}{\mathop \sum }}\,\sigma \mathbb{E}\parallel \varepsilon _{i,t}^{x}\parallel -\frac{1}{2\alpha n}\underset{t=1}{\overset{T}{\mathop \sum }}\,{{\rho }^{T-t}}\underset{i=1}{\overset{n}{\mathop \sum }}\,\mathbb{E}\parallel \varepsilon _{i,t}^{x}{{\parallel }^{2}}\nonumber\\&\le \sum\limits_{t=1}^{T}{\frac{{\alpha}n}{2}{{\sigma }^{2}}{{\rho }^{T-t}}}.
	\end{align}
	The rest proof is the same as Theorem 1. Through (38),(45), (46) and (48), we can obtain (44).
\end{proof}
\textit{Remark 6.} If $\sum\limits_{i=1}^{n}\mathbb{E}{\left({{F}_{i}^{t}}\right)}=o(1)$ and $\mathbb{E}{\left({\widetilde{{\theta}_{t}}}\right)}=o(1)$ as $t\to \infty $, then we have ${{\lim }_{T\to \infty }}\mathbb{E}\left[ R_{T}^{F} \right]\le \frac{\alpha {{\sigma }^{2}n}}{2\left( 1-\rho  \right)}+\frac{\delta L_1}{r\left( 1-\rho  \right)}+\frac{2{{M}^{2}}}{\alpha }$, which leads to ${{\lim }_{T\to \infty }}\mathbb{E}\left[ R_{T}^{F} \right]$ is a  bounded quantity.  Comparing with the Algorithm proposed in \cite{m}, we obtain that under the bandit feedback, the upper bound of $\underset{i=1}{\overset{n}{\mathop \sum }}\,\left[ {{f}_{T}}\left( x_{i}^{T} \right)-{{f}_{T}}\left( x_{*}^{T} \right) \right]$ as $T \to \infty$ will be larger, and the larger part is related to the introduction of $\delta$ due to unknown gradient information.

\section{DFFR of Distributed Online Projection-free Algorithm}
Projection-based algorithms simplify data and reduce noise, but they can also lead to information loss, particularly with high-dimensional parameters and large datasets. In contrast, projection-free algorithms avoid this complexity, simplifying updates and reducing storage and communication costs, making them more efficient for high-dimensional problems.  The Frank-Wolfe algorithm is a prominent example \cite{wenxian1,wenxian2,wenxian3,wenxian4}. Inspired by them, we propose a distributed online projection-free  algorithm and analyze its DFFR in this section. Below is the specific algorithm:
\hspace*{\fill} \\
{\noindent}	 \rule[0pt]{9cm}{0.1em}\\
\textbf{Algorithm 2}: Distributed Online Projection-free Algorithm\\
{\noindent}	 \rule[5pt]{9cm}{0.05em}\\
\textbf{Initialize}: $x_{i}^{1}\in {{X}_{\delta }}$ for any $i\in \left[ n \right]$.\\
\textbf{For} $t = 1,...,T$ \par
\textbf{for} $i = 1,...,n$ \par update the estimate as
\begin{align}
	&v_{i}^{t}=\arg {{\min }_{{{v}^{t}}\in X}}\left\langle \nabla f_{i}^{t}(x_{i}^{t}),{{v}^{t}} \right\rangle,\\
	&z_i^{t+1}=\sum_{j=1}^{n}{w_{ij}x_j^t},\\
	&\alpha _{i}^{t}=\arg {{\min }_{\alpha \in \left[ 0,1 \right]}}f_{i}^{t}(z_{i}^{t+1}+\alpha (v_{i}^{t}-x_{i}^{t})),\\
	&x_{i}^{t+1}=z_{i}^{t+1}+\alpha _{i}^{t}(v_{i}^{t}-x_{i}^{t}),
\end{align}\par
\textbf{end for}\\
\textbf{end for}\\
\textbf{Output: {$\{x_i^t\}$}}.\\
{\noindent}	 \rule[5pt]{9cm}{0.05em}

\begin{theorem}
	Suppose  Assumptions 1, 2 and 3 hold. Let ${{\nu }^{\text{t}}}=\frac{1}{n}\sum\limits_{i=1}^{n}{\left\| x_{i}^{t}-x_{*}^{t} \right\|}$ to measure the average distance between the decisions of agents and the best decision. Then for any $1 > \rho >\lambda > 0$,
	\begin{align}
		R_{T}^{F}&\le {L}\sum\limits_{t=1}^{T}{{{\rho }^{T-t}}}{\nu}_t+\sum\limits_{t=1}^{T}{{{\rho }^{T-t}}}(2L{{\alpha }_{0}}M + 2{{L}_{\text{s}}}{{\alpha }_{0}}^{2}{{M}^{2}})\nonumber\\&+8L\underset{t=1}{\mathop{\overset{T}{\mathop{\sum }}\,}}\,{{\rho }^{T-t}}{{F}_{i}^{t}}+ 4L\gamma \underset{j=1}{\mathop{\overset{n}{\mathop{\sum }}\,}}\,\parallel x_{j}^{1}\parallel \underset{t=1}{\mathop{\overset{T}{\mathop{\sum }}\,}}\,{{\rho }^{T-t}}{{\lambda }^{t-2}}\nonumber\\
		&+(\frac{9L}{n}+\frac{\frac{4L}{n}{{\rho }^{-2}}}{1-{}^{\lambda }/{}_{\rho }}\underset{t=1}{\overset{T}{\mathop )\sum }}\,{{\rho }^{T-t}}\sum\limits_{i=1}^{n}{\left\| \varepsilon _{i,t}^{x} \right\|}.		
	\end{align}
	If ${{\nu }^{\text{t}}}=o(1)$, $\sum\limits_{i=1}^{n}{F}_{i}^{t}=o(1)$ and  $\sum\limits_{i=1}^{n}{\left\| \varepsilon _{i,t}^{x} \right\|}=o(1)$ as $t\to \infty $, then ${{\lim }_{T\to \infty }} R_{T}^{F}=0$. This achieves $\underset{i=1}{\mathop{\overset{n}{\mathop{\sum }}\,}}\,\left( {{f}_{T}}\left( x_{i}^{T} \right)-{{f}_{T}}\left( x_{*}^{T} \right) \right) = 0$ as $T \to \infty$.
\end{theorem}

\begin{proof}
	Following the previous approach, let's start analyzing from $f_{i}^{t}(x_{i}^{t})-f_{i}^{t}(x_{*}^{t})$, where
	\begin{align}
		&f_i^t(x_i^t)-f_{i}^{t}(x_{*}^{t})=f_i^t(x_i^{t+1})-f_{i}^{t}(x_{*}^{t})+f_i^t(x_i^t)-f_i^t(x_i^{t+1}).
	\end{align}	
	According to $(52)$ and the convexity of function $f_{i}^{t}\left( \cdot  \right)$, we can get
	\begin{align}
		&f_i^t(x_i^t)-f_{i}^{t}(x_{*}^{t})\le f_i^t(z_i^{t+1}+{\alpha}_i^t(v_i^t-x_i^t))-f_{i}^{t}(x_{*}^{t})\nonumber\\
		&+ \left\langle \nabla f_{i}^{t}\left( x_{i}^{t} \right),x_{i}^{t}-x_{i}^{t+1} \right\rangle.
	\end{align}
	For any fix ${\alpha}_0 \in (0,1)$, it follows that
	\begin{align}
		&f_i^t(x_i^t)-f_{i}^{t}(x_{*}^{t})\le f_i^t(z_i^{t+1}+{\alpha}_0(v_i^t-x_i^t))-f_{i}^{t}(x_{*}^{t})\nonumber\\
		&+ \left\langle \nabla f_{i}^{t}\left( x_{i}^{t} \right),x_{i}^{t}-x_{i}^{t+1} \right\rangle.
	\end{align}
	By Assumption 1, we have
	\begin{align}
		&f_i^t(x_i^t)-f_{i}^{t}(x_{*}^{t})\le f_i^t(z_i^{t+1})+\left\langle \nabla f_{i}^{t}\left( z_{i}^{t+1} \right), {\alpha}_0(v_{i}^{t}-x_{i}^{t}) \right\rangle\nonumber\\
		&+\frac{L_s}{2}{\alpha}_0^2 \left\| v_{i}^{t} \right.-\left. x_{i}^{t} \right\|^2-f_i^t(x_{*}^{t})+L\left\|x_i^t\right.-\left.x_i^{t+1}\right\|\nonumber\\
		&\le f_i^t(z_i^{t+1})-f_{i}^{t}(x_{*}^{t})+L{\alpha}_0\left\|v_i^t-x_i^t\right\|+\frac{L_s}{2}{\alpha}_0^2 \left\| v_{i}^{t} \right.-\left. x_{i}^{t} \right\|^2\nonumber\\&+L\left\|x_i^t\right.-\left.x_i^{t+1}\right\|.
	\end{align}	
	By Assumption 1, 2 and (50), we can further obtain
	\begin{align}
		&f_i^t(x_i^t)-f_{i}^{t}(x_{*}^{t})\le \left\langle \nabla f_i^t\left(z_i^{t+1} \right),z_i^{t+1}-x_{*}^{t}\right\rangle+2L{\alpha}_0M\nonumber\\&+2L_sM^2{\alpha}_0^2+L\left\|x_i^t\right.-\left.x_i^{t+1}\right\|\nonumber\\
		&\le L\left\| \sum\limits_{j=1}^{n}{{{w}_{ij}}x_{j}^{t}-x_{*}^{t}} \right\|+2L{\alpha}_0M+2L_sM^2{\alpha}_0^2\nonumber\\&+L\left\|x_i^t\right.-\left.x_i^{t+1}\right\|.
	\end{align}	
	Note that the adjacent matrix $W$ is doubly stochastic, we get
	\begin{align}
		\sum\limits_{i=1}^{n}{\sum\limits_{j=1}^{n}{{{w}_{ij}}}}\left\| x_{j}^{t} \right.-{{\left. x_{*}^{t} \right\|}}=\sum\limits_{i=1}^{n}{\left\|x_{i}^{t} \right.-{{\left.  x_{*}^{t} \right\|}}}.
	\end{align}	
	Through equations $(54)-(59)$, we can obtain that
	\begin{align}
		&\frac{1}{n}\sum\limits_{i=1}^{n}{\left[ f_{i}^{t}(x_{i}^{t})-f_{i}^{t}(x_{*}^{t}) \right]}\le \frac{L}{n}\sum\limits_{i=1}^{n}{\left\| x_{i}^{t}-x_{*}^{t} \right\|}+2L{{\alpha}_{0}}M\nonumber\\&+2{{L}_{\text{s}}}{{\alpha }_{0}}^{2}{{M}^{2}}+\frac{L}{n}\sum\limits_{i=1}^{n}{\left\| x_{i}^{t}-x_{i}^{t+1} \right\|}.
	\end{align}	
	Therefore, 
	\begin{align}
		R_{T}^{F}& \le \sum\limits_{t=1}^{T}{{{\rho }^{T-t}}}\frac{1}{n}\sum\limits_{i=1}^{n}{\left[ f_{i}^{t}(x_{i}^{t})-f_{i}^{t}(x_{*}^{t}) \right]}\nonumber\\&+\frac{2L}{n}\sum\limits_{t=1}^{T}{{{\rho }^{T-t}}} \sum\limits_{i=1}^{n}{\left\| x_{i}^{t}-\overline{{{x}^{t}}} \right\|}\nonumber\\
		&\le {L}\sum\limits_{t=1}^{T}{{{\rho }^{T-t}}}{\nu}_t
		+\sum\limits_{t=1}^{T}{{{\rho }^{T-t}}}(2L{{\alpha}_{0}}M + 2{{L}_{\text{s}}}{{\alpha}_{0}}^{2}{{M}^{2}})\nonumber\\&+\frac{L}{n}\sum\limits_{t=1}^{T}{{{\rho }^{T-t}}}\sum\limits_{i=1}^{n}{\left\| x_{i}^{t}-x_{i}^{t+1} \right\|}\nonumber\\&+\frac{2L}{n}\sum\limits_{t=1}^{T}{{{\rho }^{T-t}}}\sum\limits_{i=1}^{n}{\left\| x_{i}^{t}-\overline{{{x}^{t}}} \right\|}.
	\end{align}
	Now, let's analyze the upper bound of the last two items in $(61)$,
	\begin{align}
		\left\| x_{i}^{t} \right.-\left. x_{i}^{t+1} \right\|&=\left\| x_{i}^{t}-z_{i}^{t+1}+z_{i}^{t+1}-x_{i}^{t+1} \right\|\nonumber\\
		&\le \left\| x_{i}^{t}-z_{i}^{t+1} \right\|+\left\| \varepsilon _{i,t}^{x} \right\| \nonumber\\
		&\le \left\| x_{i}^{t}-\overline{{{x}^{t}}}\,  \right\|+\left\| \overline{{{x}^{t}}}\, -z_{i}^{t+1} \right\|+\left\| \varepsilon _{i,t}^{x} \right\|.
	\end{align}
	Then, for all agents, 
	\begin{align}
		&\sum\limits_{i=1}^{n}{\left\| x_{i}^{t} \right.-\left. x_{i}^{t+1} \right\|}\le \sum\limits_{i=1}^{n}{\left\| x_{i}^{t}-\overline{{{x}^{t}}}\,  \right\|}+\sum\limits_{i=1}^{n}{\left\| \overline{{{x}^{t}}}\, -\sum\limits_{j=1}^{n}{{{w}_{ij}}x_{j}^{t}} \right\|}
		\nonumber\\&+\sum\limits_{i=1}^{n}{\left\| \varepsilon _{i,t}^{x} \right\|}.
	\end{align}	
	Noting that $\left\| \cdot  \right\|$ is convex, we have
	\begin{align}
		&\sum\limits_{i=1}^{n}{\left\| x_{i}^{t} \right.-\left. x_{i}^{t+1} \right\|}\le \sum\limits_{i=1}^{n}{\left\| x_{i}^{t}-\overline{{{x}^{t}}}\,  \right\|}\nonumber\\&+\sum\limits_{i=1}^{n}{\sum\limits_{j=1}^{n}{{{w}_{ij}}}\left\| \overline{{{x}^{t}}}\, -x_{j}^{t} \right\|}+\sum\limits_{i=1}^{n}{\left\| \varepsilon _{i,t}^{x} \right\|}\nonumber \\
		&=2\sum\limits_{i=1}^{n}{\left\| x_{i}^{t}-\overline{{{x}^{t}}}\,  \right\|}+\sum\limits_{i=1}^{n}{\left\| \varepsilon _{i,t}^{x} \right\|}.
	\end{align}	
	Then, by Lemma 3, we can obtain
	\begin{align}
		\sum\limits_{t=1}^{T}{{{\rho }^{T-t}}}\sum\limits_{i=1}^{n}{\gamma {{\lambda }^{t-2}}}\sum\limits_{j=1}^{n}{\left\| x_{j}^{1} \right\|}\nonumber\\=\sum\limits_{i=1}^{n}{\gamma }\sum\limits_{j=1}^{n}{\left\| x_{j}^{1} \right\|}\sum\limits_{t=1}^{T}{{{\rho }^{T-t}}}{{\lambda }^{t-2}}.
	\end{align}	
	For the second term of (5), 
	\begin{equation}
		\begin{aligned}
			\sum\limits_{t=1}^{T}{{{\rho }^{T-t}}}\sum\limits_{i=1}^{n}{\frac{1}{n}\sum\limits_{j=1}^{n}{\left\| \varepsilon _{j,t-1}^{x} \right\|}}=\sum\limits_{t=1}^{T}{{{\rho }^{T-t}}}\sum\limits_{j=1}^{n}{\left\| \varepsilon _{j,t-1}^{x} \right\|}\\=\sum\limits_{t=1}^{T}{{{\rho }^{T-t}}}\sum\limits_{i=1}^{n}{\left\| \varepsilon _{i,t-1}^{x} \right\|}.
		\end{aligned}
	\end{equation}	
	Noting that ${{F}^{\text{t}}}\triangleq \left| \left\| \varepsilon _{i,t}^{x} \right\|-\left\| \varepsilon _{i,t-1}^{x} \right\| \right|$, we can get
	\begin{align}
		&\sum\limits_{t=1}^{T}{{{\rho }^{T-t}}}\sum\limits_{i=1}^{n}{\left\| \varepsilon _{i,t-1}^{x} \right\|}\le \sum\limits_{t=1}^{T}{{{\rho }^{T-t}}}\sum\limits_{i=1}^{n}{\left( \left\| \varepsilon _{i,t}^{x} \right\|+{{F}_{i}^{t}} \right)}\nonumber\\
		&=\sum\limits_{t=1}^{T}{{{\rho }^{T-t}}}\sum\limits_{i=1}^{n}{\left\| \varepsilon _{i,t}^{x} \right\|+}\sum\limits_{t=1}^{T}{{{\rho }^{T-t}}}\sum\limits_{i=1}^{n}{{F}_{i}^{t}}.\
	\end{align}	
	Now, for the last term of (5), 
	\begin{align}
		&\sum\limits_{t=1}^{T}{{{\rho }^{T-t}}}\sum\limits_{s=1}^{t-2}{{{\lambda }^{t-s-2}}}\sum\limits_{j=1}^{n}{\left\| \varepsilon _{j,s}^{x} \right\|}\nonumber\\
		&=\sum\limits_{s=1}^{T-2}{\sum\limits_{t=s+2}^{T}{{{\rho }^{T-t}}}}{{\lambda }^{t-s-2}}\sum\limits_{j=1}^{n}{\left\| \varepsilon _{j,s}^{x} \right\|}\nonumber\\
		&=\sum\limits_{s=1}^{T-2}{{{\rho }^{T-t-2}}}\sum\limits_{t=s+2}^{T}{{{\rho }^{-\left( t-s-2 \right)}}}{{\lambda }^{t-s-2}}\sum\limits_{j=1}^{n}{\left\| \varepsilon _{j,s}^{x} \right\|}\nonumber\\
		&=\sum\limits_{s=1}^{T-2}{{{\rho }^{T-t-2}}}\sum\limits_{v=0}^{T-s-2}{{{\left( \frac{\lambda }{\rho } \right)}^{v}}}\sum\limits_{j=1}^{n}{\left\| \varepsilon _{j,s}^{x} \right\|}.
	\end{align}
	Owing to $\lambda <\rho $, 
	\begin{align}
		&\sum\limits_{t=1}^{T}{{{\rho }^{T-t}}}\sum\limits_{s=1}^{t-2}{{{\lambda }^{t-s-2}}}\sum\limits_{j=1}^{n}{\left\| \varepsilon _{j,s}^{x} \right\|}\nonumber\\&\le \frac{1}{1-\frac{\lambda }{\rho }}\sum\limits_{s=1}^{T-2}{{{\rho }^{T-\text{s-}2}}}\sum\limits_{j=1}^{n}{\left\| \varepsilon _{j,s}^{x} \right\|}\nonumber\\
		&=\frac{{{\rho }^{-2}}}{1-\frac{\lambda }{\rho }}\sum\limits_{t=1}^{T-2}{{{\rho }^{T-t}}}\sum\limits_{j=1}^{n}{\left\| \varepsilon _{j,t}^{x} \right\|}\nonumber\\
		&\le \frac{{{\rho }^{-2}}}{1-\frac{\lambda }{\rho }}\sum\limits_{t=1}^{T}{{{\rho }^{T-t}}}\sum\limits_{i=1}^{n}{\left\| \varepsilon _{i,t}^{x} \right\|}.
	\end{align}	
	Therefore, through $(65)-(69)$, we obtain
	\begin{align}
		&\frac{L}{n}\sum\limits_{t=1}^{T}{{{\rho }^{T-t}}}\sum\limits_{i=1}^{n}{\left\| x_{i}^{t}-x_{i}^{t+1} \right\|}\nonumber\\
		&+\frac{2L}{n}\sum\limits_{t=1}^{T}{{{\rho }^{T-t}}}\sum\limits_{i=1}^{n}{\left\| x_{i}^{t}-\overline{{{x}^{t}}} \right\|}\nonumber\\
		&\le 4L\gamma \underset{j=1}{\mathop{\overset{n}{\mathop{\sum }}\,}}\,\parallel x_{j}^{1}\parallel \underset{t=1}{\mathop{\overset{T}{\mathop{\sum }}\,}}\,{{\rho }^{T-t}}{{\lambda }^{t-2}}+\frac{8L}{n}\underset{t=1}{\mathop{\overset{T}{\mathop{\sum }}\,}}\,{{\rho }^{T-t}}\sum\limits_{i=1}^{n}{{F}_{i}^{t}}\nonumber\\&+(\frac{9L}{n}+\frac{\frac{4L}{n}{{\rho }^{-2}}}{1-{}^{\lambda }/{}_{\rho }}\underset{t=1}{\overset{T}{\mathop )\sum }}\,{{\rho }^{T-t}}\sum\limits_{i=1}^{n}{\left\| \varepsilon _{i,t}^{x} \right\|}.
	\end{align}	
	Finally, combining the results of equations $(61)$ and $(70)$, we can obtain the upper bound of $R_T^F$ as follows
	\begin{align}
		R_{T}^{F} &\le {L}\sum\limits_{t=1}^{T}{{{\rho }^{T-t}}}{\nu}_t+\sum\limits_{t=1}^{T}{{{\rho }^{T-t}}}(2L{{\alpha }_{0}}M + 2{{L}_{\text{s}}}{{\alpha }_{0}}^{2}{{M}^{2}})\nonumber\\&+\frac{8L}{n}\underset{t=1}{\mathop{\overset{T}{\mathop{\sum }}\,}}\,{{\rho }^{T-t}}\sum\limits_{i=1}^{n}{{F}_{i}^{t}}+ 4L\gamma \underset{j=1}{\mathop{\overset{n}{\mathop{\sum }}\,}}\,\parallel x_{j}^{1}\parallel \underset{t=1}{\mathop{\overset{T}{\mathop{\sum }}\,}}\,{{\rho }^{T-t}}{{\lambda }^{t-2}}\nonumber\\&+(\frac{9L}{n}+\frac{\frac{4L}{n}{{\rho }^{-2}}}{1-{}^{\lambda }/{}_{\rho }}\underset{t=1}{\overset{T}{\mathop )\sum }}\,{{\rho }^{T-t}}\sum\limits_{i=1}^{n}{\left\| \varepsilon _{i,t}^{x} \right\|}.		
	\end{align}	
	At this point, the proof of Theorem 2 has been finished.
\end{proof}
\textit{Remark 7.} Theorem 2 gives the upper bound for DFFR for the Distributed Online Projection-free Algorithm.  If $\sum\limits_{i=1}^{n}{F}_{i}^{t}=o(1)$ , $\sum\limits_{i=1}^{n}{\left\| \varepsilon _{i,t}^{x} \right\|}=o(1)$ and ${\nu}_t=o(1)$ as $t\to \infty $, then ${{\lim }_{T\to \infty }} R_{T}^{F}\le \frac{(2L{{\alpha }_{0}}M + 2{{L}_{\text{s}}}{{\alpha }_{0}}^{2}{{M}^{2}})}{1-{\rho}}$. We note that the upper bound of $R_{T}^{F}$ is a constant about ${\alpha}_0$. Since ${\alpha}_0$ is arbitrary between 0 and 1, we can obtain that ${{\lim }_{T\to \infty }} R_{T}^{F} =0$. So we can analyze that the algorithm has good tracking performance, that is, $\underset{i=1}{\mathop{\overset{n}{\mathop{\sum }}\,}}\,\left( {{f}_{T}}\left( x_{i}^{T} \right)-{{f}_{T}}\left( x_{*}^{T} \right) \right) = 0$ as $T \to \infty$.

\section{Numerical simulations}

For a tracking system where $n$ agents track $n$ targets, let ${{z}_{i}}\left( s \right)$ and $\overset{\sim}{\mathop{{{z}_{i}}}}\,\left( s \right)$ represent the position of target $i$ and agent $i$ at times $s \in (t,t+1)$, respectively.
\begin{equation}
	{{z}_{i}}\left( s \right)=\sum\limits_{k=1}^{{{d}_{i}}}{x_{i}^{t}\left[ k \right]}c_{k}^{t}\left( s \right),
\end{equation}
\begin{equation}
	\overset{\sim}{\mathop{{{z}_{i}}}}\,\left( s \right)=\sum\limits_{k=1}^{{{d}_{i}}}{\xi _{i}^{t}\left[ k \right]}c_{k}^{t}\left( s \right),
\end{equation}
where $x_i^t \in \mathbb{R}^d$ represents the coordinate vectors of the target $i$ and  $\xi _{i}^{t}\in \mathbb{R}^d$ represents the coordinate vectors of the tracker $i$ at time $t$, respectively. $c_k^t(s)$ are vector functions that characterize the space of potential paths as they evolve over time $[t,t+1]$ and satisfy
\begin{equation}
	\int_{t}^{t+1} \left\langle c_{k}^{t}(s), c_{l}^{t}(s) \right\rangle \, ds = \begin{cases} 1, & \text{if } k = l \\ 0, & \text{otherwise}. \end{cases}
\end{equation}
At each time $t$, agent exchanges information with its neighbors to choose their location. In order to be closer to the target, reducing the selection cost $\left\langle \pi _{i}^{t},x_{i}^{t} \right\rangle$ ($\pi_i^t \in \mathbb{R}_+^d$ is the price vector.) and minimizing global losses. The cost function faced by each agent is as follows
\begin{align}
	f_{i}^{t}\left( x_{i}^{t} \right)&=\xi _{i}^{1}\left\langle \pi _{i}^{t},x_{i}^{t} \right\rangle +\xi _{i}^{2}{{\int_{t}^{t+1}{\left\| {{z}_{i}}\left( s \right)-\overset{\sim}{\mathop{{{z}_{i}}}}\,\left( s \right) \right\|}}^{2}}ds\nonumber\\
	&=\xi _{i}^{1}\left\langle \pi _{i}^{t},x_{i}^{t} \right\rangle +\xi _{i}^{2}{{\left\| x_{i}^{t}-\xi _{i}^{t} \right\|}^{2}},
\end{align}
where $\xi _{i}^{1}$ and $\xi _{i}^{2}$ represent nonnegative constants that balance the two sub-objectives.  Specific model descriptions can be found in  \cite{last}. \par
In this section, we are examining a one-dimensional system, where the dimension is $d=1$, and the constraint region is defined as $X$ being the interval $\left[ -10, 10 \right]$. Set $\xi _{i}^{1} = 0$ and $\xi _{i}^{2}=1$. We set $\xi _{i}^{t}=\frac{60}{t^2}$, but in order to adapt to more complex situations, we set different position functions for different agents. Ultimately, the objective functions faced by each agent are $f_{1}^{t}\left( x_{1}^{t} \right)=\left\| x_{1}^{t}-\frac{60}{{{t}^{2}}} \right\|^2$, $f_{2}^{t}\left( x_{2}^{t} \right)=\left\| 2x_{2}^{t}-\frac{60}{{{t}^{2}}} \right\|^2$, $f_{3}^{t}\left( x_{3}^{t} \right)=\left\| 3x_{3}^{t}-\frac{60}{{{t}^{2}}} \right\|^2$ and $f_{4}^{t}\left( x_{4}^{t} \right)=\left\| 6x_{4}^{t}-\frac{60}{{{t}^{2}}} \right\|^2$. We set up four agents and built a undirected connected network, in which the non-zero element with the smallest adjacency matrix is $0.22$, as shown in Figure 1. 
\begin{figure}[h]
	\centering
	\includegraphics[width=3.2in]{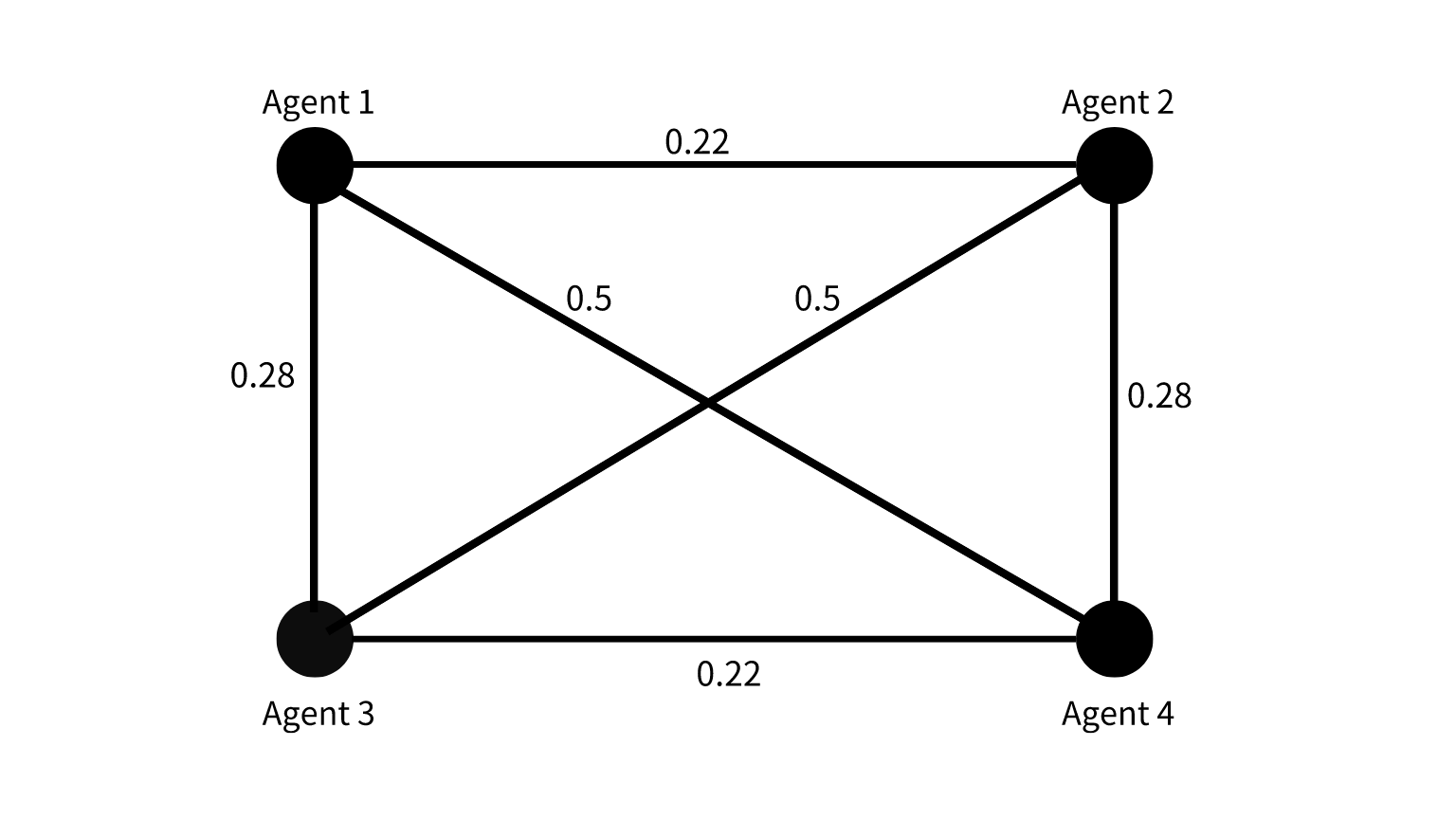}
	\caption{The topological structure of undirected graphs.}
	\label{fig1}
\end{figure}
So the $\lambda $ as defined in (3) is $0.98625$, we chose $\rho =0.9875$ and $B=1$. At the same time, we set the ${\alpha}_t=\frac{2}{t^{\frac{1}{2}}}$. Under Algorithm 1, we choose ${\delta}=0.01$. Under Algorithm 2, we set a fixed step size between 0 and 1 in advance as $0.002$. Therefore, the conditions of Theorem $1-2$ are satisfied. 
We consider the condition for the consistency of agent decision is that the difference between the maximum decision and the minimum decision is less than 0.001. Although facing different time-varying objective functions, the decisions generated by the proposed algorithms reached agreement at $t$=21 and $t$=15 respectively at a very fast speed, as shown in Figure 2 and 3.

\begin{figure}[h]
	\centering
	\includegraphics[width=3.2in]{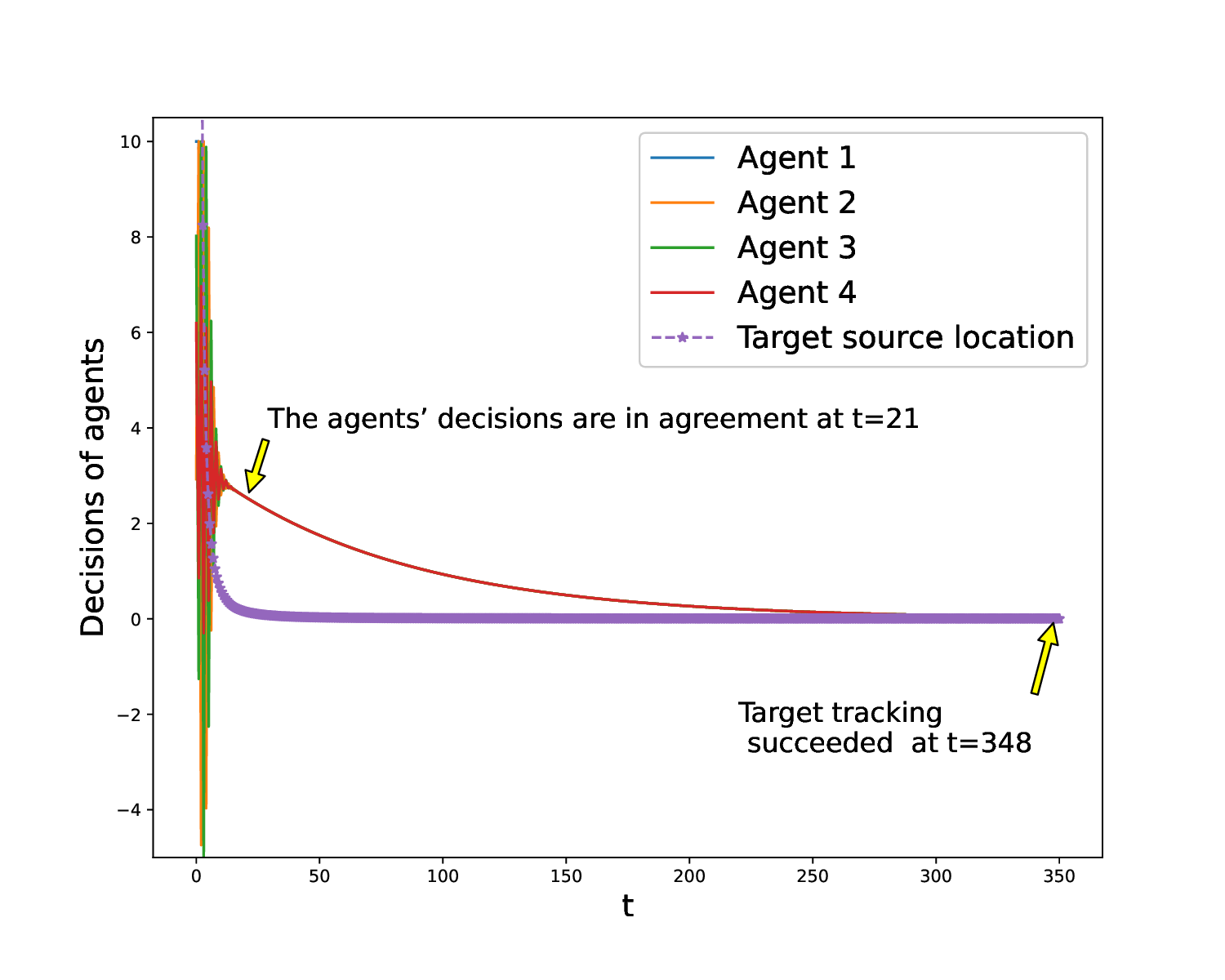}
	\caption{The decisions of Algorithm 1.}
	\label{fig1}
\end{figure}
 
\begin{figure}[h]
 	\centering
 	\includegraphics[width=2.6in]{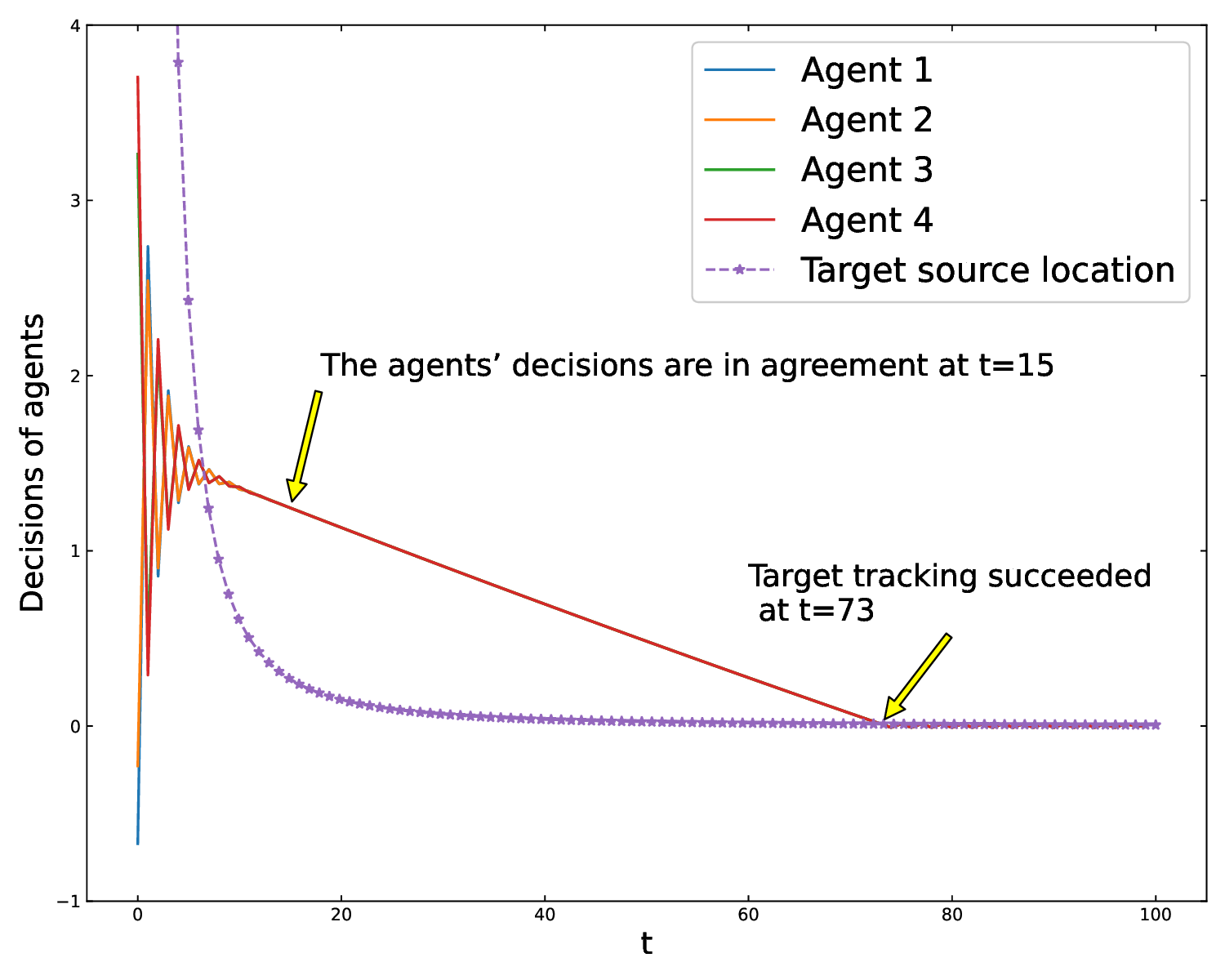}
 	\caption{The decisions of Algorithm 2.}
 	\label{fig1}
\end{figure}

Assuming that the difference between the agent's decision and the target is less than 0.001, we determine that the target tracking is successful. Through comparison, it can be found that Algorithm 2 is faster than Algorithm 1 in tracking the target. Algorithm 1 needs at least 348 iterations, while Algorithm 2 only needs 73 iterations.  

The DFFR of Algorithm 1-2 show good convergence, as shown in Figure 4. Compared with Distributed Online Gradient Decent Algorithm, Algorithm 2 has similar convergence and tracking performance, while Algorithm 1 has a slower convergence speed due to the lack of the information of gradient, as shown in Figure 4 and 5.
\begin{figure}[h]
	\centering
	\includegraphics[width=3.2in]{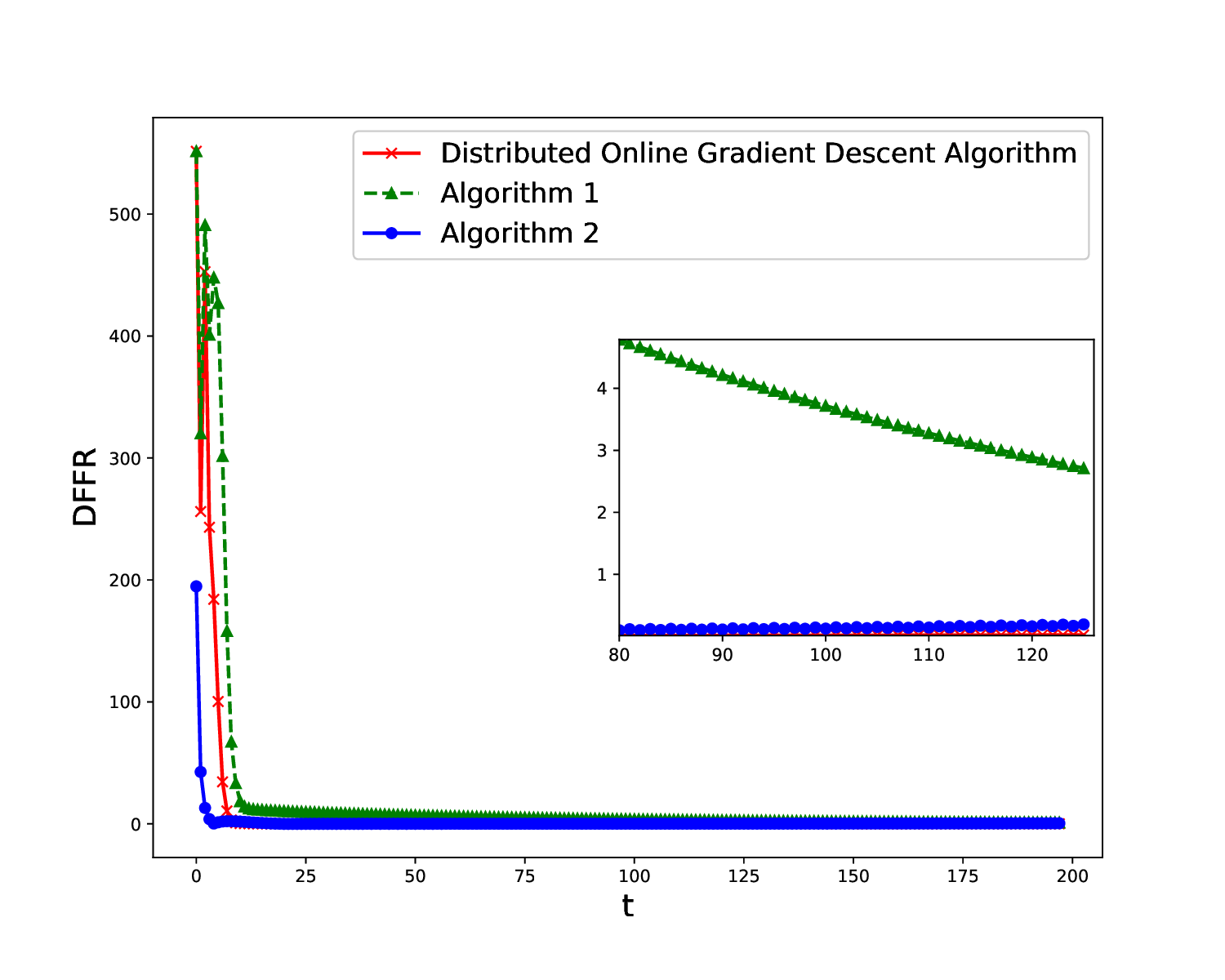}
	\caption{DFFR of Algorithms.}
	\label{fig1}
\end{figure}
\begin{figure}[h]
	\centering
	\includegraphics[width=3.2in]{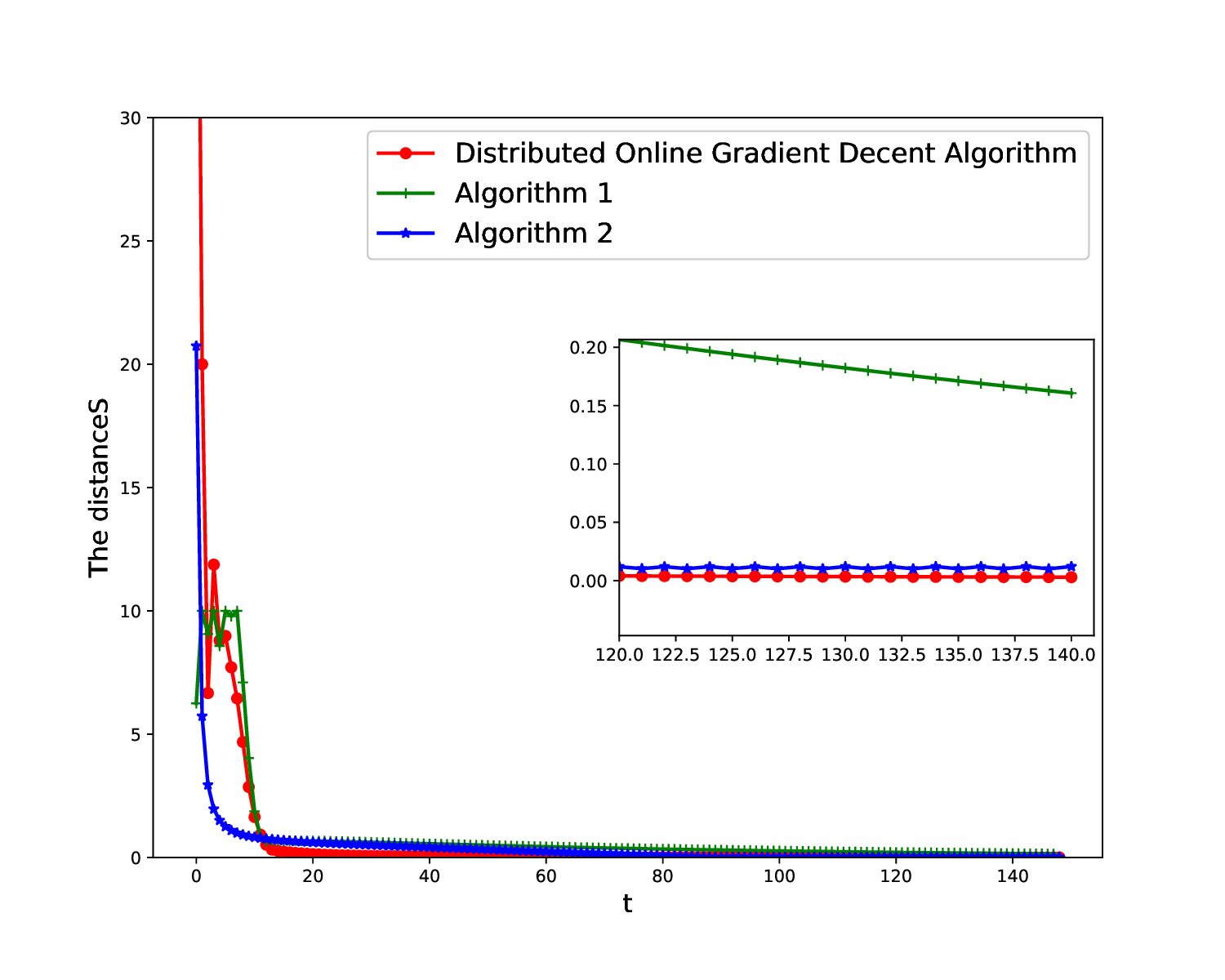}
	\caption{$\frac{1}{n}\sum\limits_{i=1}^{n}{\left\| x_{i}^{t}-x_{*}^{t} \right\|}$.}
	\label{fig1}
\end{figure}

Meanwhile, for Distributed Online Gradient Decent Algorithm, we also compared its convergence under classical regret and DFFR conditions, as shown in Figure 6. Through DFFR, we can know the minimum time for the agent decisions to converge to the optimal decision.
If the tracking error is less than 0.01, the minimum time required to meet this tracking condition increases as $\rho$ increases, and it is the largest when $\rho$ is not present.
By comparison, we can obtain that DFFR has better tracking performance.

\begin{figure}[!tb]
	\centering
	\includegraphics[width=3.2in]{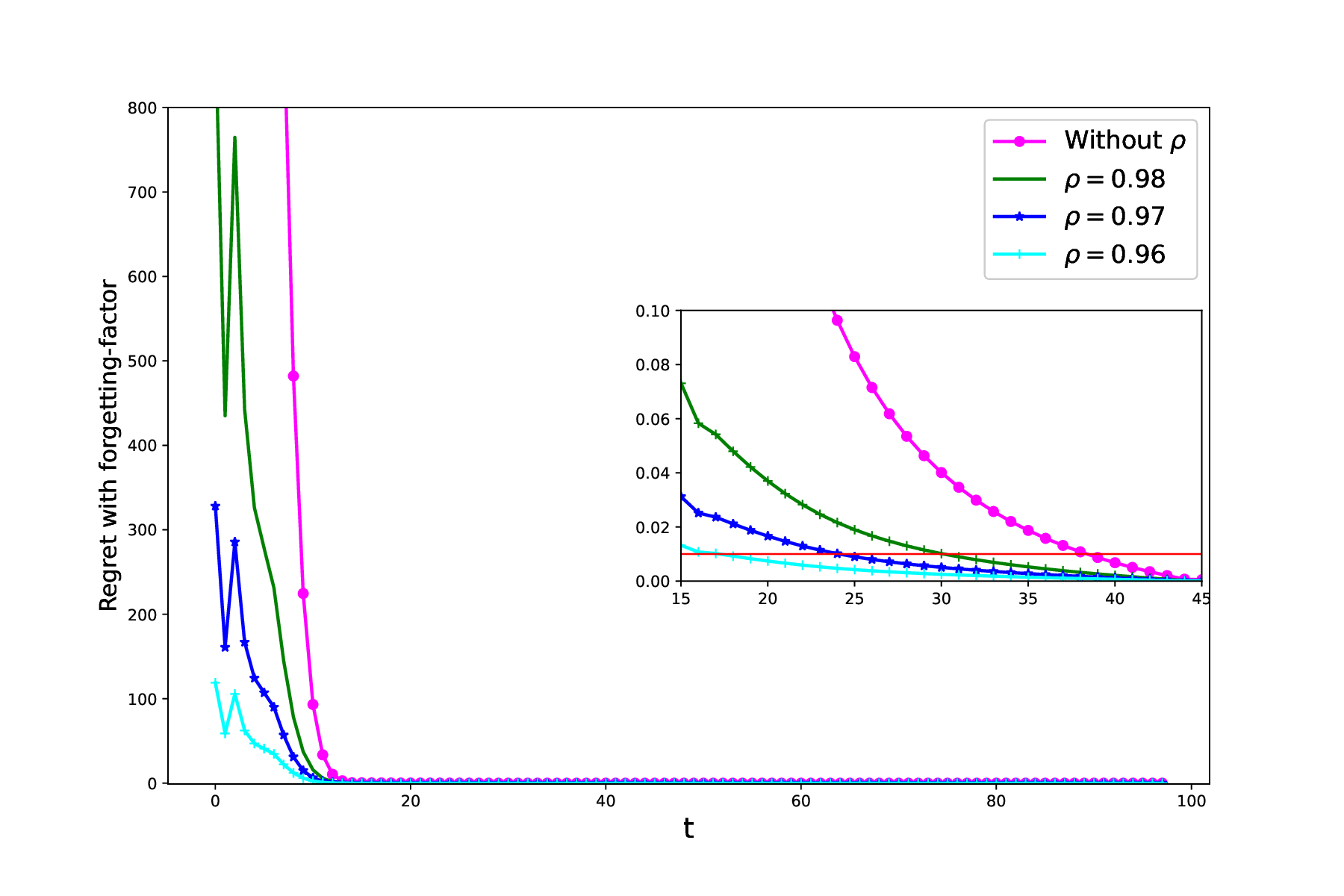}
	\caption{DFFR of Distributed Online Gradient Decent Algorithm with ${\rho}=0.98, {\rho}=0.97, {\rho}=0.96$ and without ${\rho}$.}
	\label{fig5}
\end{figure}

\section{Conclusions}\label{sec5}
In this paper, we propose a new metric for measuring distributed online optimization algorithms, namely the distributed forgetting-factor regret (DFFR). Firstly, we propose Distributed Online Gradient-free Algorithm and Distributed Online Projection-free Algorithm. Then, we derive the upper bounds of the DFFR and provide mild conditions for the regret to approach zero or to have a tight bound as $T\to \infty $. Finally, we validate the effectiveness of the algorithms and the superior performance of DFFR through experimental simulation.  In the future, we will extend our results to both direct and switching topology scenarios.\par

\section*{Acknowledgments}
 This work is supported by the National Natural Science Foundation of China under Grant No. 62473009.

\end{document}